\documentclass[10pt,draftcls,,onecolumn]{IEEEtran}

\usepackage{graphicx}
\usepackage{amsmath,epsfig,amssymb, amstext, verbatim,amsopn,cite,subfigure,multirow,multicol,lipsum}
\usepackage{url}
\usepackage{amsfonts}
\usepackage{epsfig}
\usepackage{epstopdf}
\usepackage{setspace}
\usepackage{stmaryrd}
\usepackage{psfrag}
\allowdisplaybreaks

\newcommand{\Equal}{\hspace{-0.7mm}=\hspace{-0.7mm}}
\newcommand{\Add}{\hspace{-0.7mm}+\hspace{-0.7mm}}
\newcommand{\Minus}{\hspace{-0.7mm}-\hspace{-0.7mm}}
\newcommand{\Less}{\hspace{-0.7mm}<\hspace{-0.7mm}}
\newcommand{\Great}{\hspace{-0.7mm}>\hspace{-0.7mm}}

\newtheorem{theorem}{Theorem}

\newenvironment{proof}[1][Proof]{\begin{trivlist}
\item[\hskip \labelsep {\bfseries #1}]}{\end{trivlist}}

\newenvironment{remark}[1][Remark:]{\begin{trivlist}
\item[\hskip \labelsep {\bfseries #1}]}{\end{trivlist}}

\newcommand{\qed}{\nobreak \ifvmode \relax \else
      \ifdim\lastskip<1.5em \hskip-\lastskip
      \hskip1.5em plus0em minus0.5em \fi \nobreak
      \vrule height0.75em width0.5em depth0.25em\fi}

\ifCLASSINFOpdf
\else
\fi

\begin{document}

\title{Adaptive Mode Selection in Bidirectional Buffer-aided Relay Networks with Fixed Transmit Powers 
}
\author{Vahid Jamali$^\dag$, Nikola Zlatanov$^\ddag$, Aissa Ikhlef$^\ddag$, and Robert Schober$^\dag$ \\
\IEEEauthorblockA{$^\dag$ Friedrich-Alexander University (FAU), Erlangen, Germany \\
 $^\ddag$ University of British Columbia (UBC), Vancouver, Canada}
}

\maketitle

\begin{abstract} 
We consider a bidirectional network in which two users exchange information with the help of a buffer-aided relay. In such a network without direct link between user 1 and user 2, there exist six possible transmission modes, i.e., four point-to-point modes (user 1-to-relay, user 2-to-relay, relay-to-user 1,  relay-to-user 2), a multiple access mode (both users to the relay), and  a broadcast mode (the relay to both users). Because of the buffering capability at the relay, the transmissions in the network are not restricted to adhere to a predefined schedule,  and  therefore, all the transmission modes in the bidirectional relay network can be used adaptively  based on the instantaneous channel state information (CSI) of the involved links. For the considered network, assuming fixed transmit powers for both the users and the relay, we derive the optimal transmission mode selection policy which maximizes the sum rate. The proposed policy selects one out of the six possible transmission modes in each time slot based on the instantaneous CSI. Simulation results confirm the effectiveness of the proposed protocol compared to existing protocols. 
\end{abstract}
%

\section{Introduction} \label{Sec I (Intro)}

In a bidirectional relay network, two users  share information via a relay node \cite{Tarokh}. Several protocols have been proposed for such a network under the practical half-duplex constraint, i.e., a node cannot transmit and receive at the same time and in the same frequency band. The simplest protocol is the traditional two-way protocol in which the transmission is accomplished in four successive point-to-point phases: user 1-to-relay, relay-to-user 2, user 2-to-relay, and relay-to-user 1. In contrast, the time division broadcast (TDBC) protocol exploits the broadcasting capability of the wireless medium  and combines the relay-to-user 1 and relay-to-user 2 phases into one phase, the broadcast phase \cite{TDBC}. Thereby, the relay broadcasts a superimposed codeword, carrying information for both user 1 and user 2, such that each user is able to recover its intended information by self-interference cancellation. To further increase the spectral efficiency, the multiple access broadcast (MABC) protocol was proposed in which the user 1-to-relay and user 2-to-relay phases are also combined into one phase, the  multiple-access phase \cite{MABC}. In the  multiple-access phase, user 1 and user 2 simultaneously transmit to the relay, which decodes both messages. For the bidirectional relay network without a direct link between user 1 and user 2, the capacity regions of all six possible transmission modes, i.e., the four point-to-point modes, the multiple access mode, and the broadcast mode, are known  \cite{BocheIT},\cite{Cover}. Using this knowledge,  a significant research effort has been dedicated to obtaining the achievable rate region and the capacity of the bidirectional relay network with and without fading
 \nocite{Tarokh,BocheIT,BochePIMRC,PopovskiICC,PopovskiLetter} \cite{Tarokh}-\cite{ PopovskiLetter}. Using this knowledge,  a significant research effort has been dedicated to obtaining the achievable rate region  of the bidirectional relay network
 \nocite{Tarokh,BocheIT,BochePIMRC,PopovskiICC,PopovskiLetter} \cite{Tarokh}-\cite{ PopovskiLetter}.  Specifically, the achievable rates of  most existing protocols for two-hop relay transmission are
limited by the instantaneous capacity of the weakest link associated with the relay.  The reason for this is the fixed schedule of using the transmission modes which is adoped in all existing protocols, and does not exploit the instantaneous channel state information (CSI) of the involved links.   For one-way relaying, an adaptive link selection protocol  was proposed   in \cite{NikolaJSAC} where based on the instantaneous CSI, in each time slot, either the source-relay link or the relay-destination link is selected for transmission. To this end, the relay has to have a buffer for data storage. This strategy was shown to achieve the capacity of the one-way relay channel with fading \cite{NikolaTIT}. 

Motivated by the   protocols in \cite{NikolaJSAC} and \cite{NikolaTIT}, our goal is to utilize all available degrees of freedom in the three-node half-duplex bidirectional relay network with fading, via a buffer-aided and adaptive mode selection protocol. In particular,   given the fixed transmit powers of all three nodes in the bidirectional relay network,  we find a policy which selects the optimal transmission mode from the six available modes in each time slot such that the sum rate is maximized. A similar problem was considered in \cite{PopovskiLetter}. However, the selection policy in \cite{PopovskiLetter} does not use all possible modes, i.e., it only selects between two point-to-point modes and the broadcast mode, and assumes that the transmit powers of all three nodes are identical. We will show that considering only the three modes in \cite{PopovskiLetter} is not optimal. In fact, the multiple access mode is selected  with non-zero probability in the proposed optimal selection policy. Finally, we note that the advantages of buffering come at the expense of an increased end-to-end delay. Therefore, the proposed protocol is suitable for applications that are not sensitive to delay. The delay analysis of the proposed protocol is beyond the scope of the current work and is left for future research.

The  remainder of  the  paper  is  organized as  follows. In Section \ref{SysMod}, the system model is presented. In Section \ref{AdapMod}, the average sum rate is investigated and the optimal mode selection policy is provided. Simulation results are presented in Section \ref{SimRes}. Finally, Section \ref{Conclusion} concludes the paper.

\section{System Model}\label{SysMod}
In this section, we first describe the channel model and review the achievable rates for the six possible  transmission modes.

\subsection{Channel Model}
\begin{figure}
\centering
\psfrag{U1}[c][c][0.75]{$\text{User 1}$}
\psfrag{U2}[c][c][0.75]{$\text{User 2}$}
\psfrag{R}[c][c][0.75]{$\text{Relay}$}
\psfrag{h1}[c][c][0.75]{$h_1(i)$}
\psfrag{h2}[c][c][0.75]{$h_2(i)$}
\psfrag{P1}[c][c][0.75]{$P_1(i)$}
\psfrag{P2}[c][c][0.75]{$P_2(i)$}
\psfrag{Pr}[c][c][0.75]{$P_r(i)$}
\psfrag{B1}[c][c][0.75]{$B_1$}
\psfrag{B2}[c][c][0.75]{$B_2$}
\includegraphics[width=3.5 in]{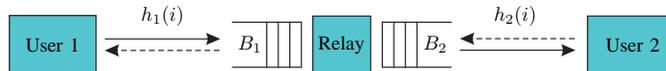}
\caption{Three-node bidirectional relay network consisting of two users and a relay.}
\label{FigSysMod}
\end{figure}
\begin{figure}
\centering
\psfrag{U1}[c][c][0.5]{$\text{User 1}$}
\psfrag{U2}[c][c][0.5]{$\text{User 2}$}
\psfrag{R}[c][c][0.5]{$\text{Relay}$}
\psfrag{M1}[c][c][0.75]{$\mathcal{M}_1$}
\psfrag{M2}[c][c][0.75]{$\mathcal{M}_2$}
\psfrag{M3}[c][c][0.75]{$\mathcal{M}_3$}
\psfrag{M4}[c][c][0.75]{$\mathcal{M}_4$}
\psfrag{M5}[c][c][0.75]{$\mathcal{M}_5$}
\psfrag{M6}[c][c][0.75]{$\mathcal{M}_6$}
\includegraphics[width=3.4 in]{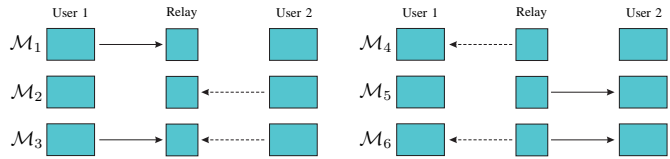}
\caption{The six possible transmission modes in the considered bidirectional relay network.}
\label{FigModes}
\end{figure}

We consider a simple network in which users 1 and 2 exchange  information with the help of a relay node, as shown in Fig. \ref{FigSysMod}. We assume that there is no direct link between user 1 and user 2, and thus, user 1 and user 2
communicate with each other only through the relay node. We assume that all three nodes in the network are half-duplex. Furthermore, we assume that  time is divided into slots of equal length and each node transmits codewords which  span one time slot except in the multiple access mode in which the nodes may transmit codewords which span a fraction of the time slot. The duration of one time slot is normalized to one. We assume that the  users-to-relay and relay-to-users channels are all impaired by AWGN having zero mean and unit variance, and block fading, i.e., the channel coefficients are constant during one time slot and change from one time slot to the next. Moreover, the channel coefficients are assumed to be reciprocal such that the user 1-to-relay and user 2-to-relay  channels are identical to the relay-to-user 1 and relay-to-user 2 channels in each time slot, respectively. Let $h_1(i)$ and $h_2(i)$ denote the channel coefficients between user 1 and the relay and between user 2 and the relay in the $i$-th time slot, respectively. Furthermore, let $S_1(i)=|h_1(i)|^2$ and $S_2(i)=|h_2(i)|^2$  denote the squares  of the channel coefficient amplitudes in the $i$-th time slot. $S_1(i)$ and $S_2(i)$   are assumed to be ergodic and stationary random processes with means $\Omega_1=E\{S_1\}$ and  $\Omega_2=E\{S_2\}$\footnote{In the paper, we drop time index $i$ in expectations for notational simplicity.}, respectively, where $E\{\cdot\}$ denotes expectation.  Since the noise is AWGN, in order to achieve the capacity in each mode, each  node has to  transmit codewords which are  Gaussian distributed. Therefore,  the transmitted codewords of user 1, user 2, and the relay are composed of symbols which are modeled as Gaussian   random variables with variances $P_1, P_2$, and $P_r$, respectively, where $P_j$ is the transmit power of node $j\in\{1,2,r\}$. For ease of notation, we define $C(x)\triangleq \log_2(1+x)$.  In the following, we introduce the transmission modes and their achievable rates.

\subsection{Transmission Modes and Their Achievable  Rates}

In the considered bidirectional relay network only six transmission modes are possible, cf. Fig. \ref{FigModes}. The six possible transmission modes are denoted by  ${\cal M}_1,...,{\cal M}_6$, and  $R_{jj'}(i)\geq 0, \,\,j,j'\in\{1,2,r\}$, denotes the transmission rate from node $j$ to node $j'$ in the $i$-th time slot. Let $B_1$ and $B_2$ denote two infinite-size buffers at the relay in which the received information from user 1 and user 2 is stored, respectively. Moreover, $Q_j(i), \,\, j\in\{1,2\}$ denotes as the amount of normalized information in bits/symbol available in buffer $B_j$ in the $i$-th time slot. Using this notation, the transmission modes and their respective rates are presented in the following:

${\cal M}_1$: User 1 transmits to the relay and user 2 is silent. In this mode,     the maximum rate from user 1 to the relay in  the $i$-th time slot is given by $R_{1r}(i)= C_{1r}(i)$, where $C_{1r}(i)=C(P_1S_1(i))$. The relay decodes this information and stores it in buffer $B_1$. Therefore, the amount of information in buffer $B_1$ increases to $Q_1(i)=Q_1(i-1)+R_{1r}(i)$.

${\cal M}_2$: User 2 transmits to the relay and user 1 is silent. In this mode,    the maximum rate from user 2 to the relay in  the $i$-th time slot is given by $R_{2r}(i)= C_{2r}(i)$, where $C_{2r}(i)=C(P_2S_2(i))$.  The relay decodes this information and stores it in buffer $B_2$. Therefore, the amount of information in buffer $B_2$ increases to $Q_2(i)=Q_2(i-1)+R_{2r}(i)$.

${\cal M}_3$: Both users 1 and 2 transmit to the relay  simultaneously. For this mode, we assume that multiple access transmission is used, see \cite{Cover}. Thereby, the maximum achievable sum rate in the $i$-th time slot is given by $R_{1r}(i)+R_{2r}(i)= C_{r}(i)$, where $C_{r}(i)=C(P_1S_1(i)+P_2S_2(i))$.  
 Since user 1 and user 2 transmit independent messages, the sum rate, $C_r(i)$, can be decomposed into two rates, one from user 1 to the relay and the other one from user 2 to the relay. Moreover, these two capacity rates can be achieved via time sharing and successive interference cancellation. Thereby, in the first $0\leq t(i)\leq 1$ fraction of the $i$-th time slot, the relay first decodes the
codeword received from user 2 and considers the signal from user 1 as noise. Then, the relay subtracts the signal component of user 2 from the received signal and decodes the
 codeword received from user 1. A similar procedure is performed in the remaining $1-t(i)$ fraction of the $i$-th time slot but now the relay first
decodes the codeword from user 1 and treats the signal of user 2 as noise, and then decodes the codeword received from user 2. Therefore, for a given $t(i)$, we decompose $C_r(i)$ as $C_r(i)=C_{12r}(i)+C_{21r}(i)$ and  the maximum rates from users 1 and 2 to the relay in the $i$-th time slot are $R_{1r}(i) = C_{12r}(i)$ and $R_{2r}(i)=C_{21r}(i)$, respectively. $C_{12r}(i)$ and $C_{21r}(i)$ are given by 
\begin{IEEEeqnarray}{ll}\label{Cr}
  C_{12r}(i)\hspace{-0.5mm}=\hspace{-0.5mm}t(i) C\left(P_1S_{1}(i)\right)\hspace{-1mm} + \hspace{-1mm}(1-t(i)) C\left(\frac{P_1S_{1}(i)}{1+P_2S_{2}(i)}\right)  \IEEEeqnarraynumspace\IEEEyesnumber\IEEEyessubnumber \\
   C_{21r}(i)\hspace{-0.5mm}=\hspace{-0.5mm}(1-t(i)) C\left(P_2S_{2}(i)\right) \hspace{-1mm} +  \hspace{-1mm} t(i) C\left(\frac{P_2S_{2}(i)}{1+P_1S_{1}(i)}\right) \IEEEyessubnumber 
\end{IEEEeqnarray}
The relay decodes the information received from user 1  and user 2  and stores it in its buffers $B_1$ and $B_2$, respectively. Therefore, the amount of information in buffers $B_1$ and $B_2$ increase to $Q_1(i)=Q_1(i-1)+R_{1r}(i)$ and $Q_2(i)=Q_2(i-1)+R_{2r}(i)$, respectively.

${\cal M}_4$: The relay transmits the information received  from user 2 to user 1. Specifically, the relay extracts the information from buffer $B_2$, encodes it into a codeword, and transmits it to user 1. Therefore, the transmission rate from the relay to user 1 in the $i$-th time slot is limited by both the capacity of the relay-to-user 1 channel and the amount of  information  stored in buffer $B_2$. Thus, the maximum transmission rate from the relay to user 1 is given by $R_{r1}(i) \Equal \min\{C_{r1}(i),Q_2(i-1)\}$, where $C_{r1}(i)=C(P_rS_1(i))$. Therefore, the amount of information in buffer $B_2$ decreases to $Q_2(i)\Equal Q_2(i\Minus 1)\Minus R_{r1}(i)$.

${\cal M}_5$: This mode is identical to ${\cal M}_4$ with user 1 and 2 switching places. The maximum transmission rate from the relay to user 2 is given by $R_{r2}(i) = \min\{C_{r2}(i),Q_1(i-1)\}$, where $C_{r2}(i)=C(P_rS_2(i))$ and the amount of information in buffer $B_1$ decreases to $Q_1(i)\Equal Q_1(i\Minus 1)\Minus R_{r2}(i)$.

${\cal M}_6$: The relay broadcasts to both user 1 and user 2 the information received from user 2 and user 1, respectively. Specifically, the relay extracts the information intended for user 2 from buffer $B_1$ and  the information intended for user 1 from buffer $B_2$. Then, based on the scheme in \cite{BocheIT}, it constructs a superimposed codeword which contains the information from both users and broadcasts it to both users. Thus, in the $i$-th time slot, the maximum rates from the relay to users 1 and 2 are given by  $R_{r1}(i) \hspace{-1mm}=\hspace{-1mm} \min\hspace{-0.2mm}\{C_{r1}(i),\hspace{-0.5mm}Q_2(i-1)\}$ and $R_{r2}(i)=\min\{C_{r2}(i),Q_1(i-1)\}$, respectively. Therefore, the amount of information in buffer $B_1$ and $B_2$ decrease to $Q_1(i)\Equal Q_1(i\Minus 1)\Minus R_{r2}(i)$ and $Q_2(i)\Equal Q_2(i\Minus 1)\Minus R_{r1}(i)$, respectively.

%
Our aim is to develop an  optimal selection policy which selects one of the six possible transmission modes in each time slot such that, given $P_1$, $P_2$, and $P_r$, the average sum rate of both users is maximized. 
To this end,  we introduce six binary variables  $q_k(i) \in\{0,1\}, \,\,k=1,...,6$, where  $q_k(i)$ indicates whether or not transmission mode  $\mathcal{M}_k$ is selected in the $i$-th time slot. In particular, $q_k(i)=1$ if  mode $\mathcal{M}_k$ is selected and $q_k(i)=0$ if it is not selected in the $i$-th time slot. Moreover, since in each time slot  only one of the six transmission modes can be selected, only one of the mode selection variables is equal to one and the others are zero, i.e., $\mathop \sum_{k = 1}^6 q_k(i)=1$ holds. 

 In the proposed framework, we assume that one node (e.g. the relay node) is responsible for deciding which transmission mode is selected based on the full CSI of both links and the proposed protocol, cf. Theorem~\ref{FixProt}. Then, it sends its decision to the other nodes and transmission in time slot $i$ begins. Furthermore, for each node to be able  to decode their intended messages  and adapt their transmission rates, all three nodes have to have full   CSI of both links.

\section{Adaptive Mode Selection}\label{AdapMod}

In this section, we first investigate the achievable average sum rate  of the network. Then, we formulate a maximization problem whose solution yields  the maximum average sum rate.

\subsection{Achievable Average Sum Rate}

We assume that user 1 and user 2 always have enough information to send in all time slots and that the number of time slots, $N$, satisfies $N\to \infty$. Therefore, using $q_k(i)$, the user 1-to-relay, user 2-to-relay, relay-to-user 1, and relay-to-user 2 average transmission rates, denoted by $\bar{R}_{1r}$, $\bar{R}_{2r}$, $\bar{R}_{r1}$, and $\bar{R}_{r2}$, respectively, are obtained as
\begin{IEEEeqnarray}{Cll}\label{RatReg123}
    \bar{R}_{1r} &= \underset{N\to \infty}{\lim} \frac{1}{N}\mathop \sum \limits_{i = 1}^N \left[ q_1(i)C_{1r}(i)+q_3(i)C_{12r}(i)\right] \IEEEyesnumber\IEEEyessubnumber \\
		\bar{R}_{2r} &=  \underset{N\to \infty}{\lim} \frac{1}{N}\mathop \sum \limits_{i = 1}^N \left[ q_2(i)C_{2r}(i)+q_3(i)C_{21r}(i)\right] \IEEEyessubnumber\\
    \bar{R}_{r1}\hspace{-1mm}  &=\hspace{-1mm} \underset{N\to \infty}{\lim}\frac{1}{N}\hspace{-1mm} \mathop \sum \limits_{i = 1}^N \left[ q_4(i)\hspace{-0.5mm} +\hspace{-0.5mm} q_6(i)\right]\hspace{-0.5mm} \min\{C_{r1}(i),Q_2(i-1)\} \hspace{-8mm} \IEEEeqnarraynumspace\IEEEyessubnumber \\
		\bar{R}_{r2}\hspace{-1mm}  &=\hspace{-1mm}   \underset{N\to \infty}{\lim}\frac{1}{N}\hspace{-1mm} \mathop \sum \limits_{i = 1}^N \left[ q_5(i)\hspace{-0.5mm} +\hspace{-0.5mm} q_6(i)\right]\hspace{-0.5mm} \min\{C_{r2}(i),Q_1(i-1)\}.  \,\, \hspace{-10mm}\IEEEeqnarraynumspace \IEEEyessubnumber
\end{IEEEeqnarray}
The average rate from user 1 to user 2 is the average rate that user 2 receives from the relay, i.e.,  $\bar{R}_{r2}$. Similarly, the average rate from user 2 to user 1 is the average rate that user 1 receives from the relay, i.e.,  $\bar{R}_{r1}$.
 In the following theorem, we introduce a useful condition for the queues in the buffers of the relay  leading to the optimal mode selection policy. 

\begin{theorem}[Optimal Queue Condition]\label{Queue}
\normalfont
$\,\,\,\, \text{The} \quad \text{maximum}$ sum rate, $\bar{R}_{r1}+\bar{R}_{r2}$, for the considered bidirectional relay network is obtained when  the queues in buffers $B_1$ and $B_2$ are at the edge of non-absorbtion. More precisely, the following conditions must hold for the maximum sum rate: 
\begin{IEEEeqnarray}{lll}\label{RatRegApp456-buffer} 
\bar{R}_{1r}=\bar{R}_{r2} = \underset{N\to \infty}{\lim}\frac{1}{N}\mathop \sum \limits_{i = 1}^N \left[ q_5(i)+q_6(i)\right]C_{r2}(i)  \IEEEyesnumber\IEEEeqnarraynumspace \IEEEyessubnumber \\  
\bar{R}_{2r}=\bar{R}_{r1}= \underset{N\to \infty}{\lim}\frac{1}{N}\mathop \sum \limits_{i = 1}^N \left[ q_4(i)+q_6(i)\right]C_{r1}(i) \IEEEeqnarraynumspace\IEEEyessubnumber 
\end{IEEEeqnarray}
where $\bar{R}_{1r}$ and $\bar{R}_{2r}$ are given by (\ref{RatReg123}a) and (\ref{RatReg123}b), respectively.
\end{theorem}

\begin{proof}
Please refer to Appendix \ref{AppQueue}. 
\end{proof}
Using this theorem,  we now derive the optimal transmission mode selection policy.

\subsection{Optimal Mode Selection Protocol}

The available degrees of freedom in the considered network are the mode selection variables $q_{k}(i)$, $\forall  k,i$ and the time sharing variable in the multiple access mode, $t(i)$,  $\forall  i$. In the following, we   formulate a problem for optimization of $q_k(i),\,\,\forall k,i$ and $t(i),\,\,\forall i$ such that the average sum rate, $\bar{R}_{r1}+\bar{R}_{r2}$,   is maximized:
\begin{IEEEeqnarray}{Cll}\label{FixProb}
    {\underset{q_k(i),t(i),\,\,\forall i,k}{\mathrm{maximize}}}\,\, &\bar{R}_{1r}+\bar{R}_{2r} \nonumber \\
    \mathrm{subject\,\, to} \,\, &\mathrm{C1}:\quad \bar{R}_{1r}=\bar{R}_{r2}  \nonumber \\
    &\mathrm{C2}:\quad \bar{R}_{2r}=\bar{R}_{r1} \nonumber \\
		&\mathrm{C3}:\quad \mathop \sum \nolimits_{k = 1}^6 {q_k}\left( i \right) = 1, \,\, \forall i   \nonumber \\
    &\mathrm{C4}:\quad q_k(i) [1-q_k(i)] = 0, \,\, \forall i, k \nonumber \\
    &\mathrm{C5}:\quad 0\leq t(i)\leq 1, \,\, \forall i. \IEEEyesnumber
\end{IEEEeqnarray}
where $\mathrm{C1}$ and $\mathrm{C2}$ are due to Theorem~\ref{Queue}, $\mathrm{C3}$ and $\mathrm{C4}$ guarantee that only one of the transmission modes is selected in each time slot, and $\mathrm{C5}$ specifies the interval of the time sharing variable. Considering constraints $\mathrm{C1}$ and $\mathrm{C2}$, the sum rate, $\bar{R}_{r1}+\bar{R}_{r2}$, is also given by $\bar{R}_{1r}+\bar{R}_{2r}$. 

As will be seen later, the solution of (\ref{FixProb}), i.e., the optimal mode selection policy may require coin flips. Therefore, we define $X_n(i)\in \{\mathrm{0,1}\}$ as the   outcomes of the $n$-th coin flip in the $i$-th time slot. The probabilities of the possible outcomes of the $n$-th coin flip are $\Pr\{X_n(i)=1\}=p_n$ and $\Pr\{X_n(i)=0\}=1-p_n$. In the following theorem, we provide the solution to the maximization problem in (\ref{FixProb}). To this end, we define three mutually exclusive signal-to-noise ratio (SNR) regions and each region  requires  a different optimal selection policy. For  given $\Omega_1$, $\Omega_2$, $P_1$, $P_2$, and $P_r$,  exactly one  of the  SNR regions is applicable.

\begin{theorem}[Mode Selection Policy]\label{FixProt}
\normalfont The optimal mode selection policy which maximizes the average sum rate of the considered three-node half-duplex bidirectional relay network
with AWGN and block fading is given by
\begin{IEEEeqnarray}{lll}
q_{k^*}(i)=
   \begin{cases}
     1, &k^*= {\underset{k=1,\dots,6}{\arg \, \max}} \{\mathcal{I}_k(i)\Lambda_k(i)\} \\
     0, &\mathrm{otherwise}
\end{cases}
 \IEEEyesnumber
\end{IEEEeqnarray}
where $\Lambda_k(i)$ is referred to as the selection metric,  given by
\begin{IEEEeqnarray}{lll}\label{SelecMet}
    \Lambda_1(i) = (1-\mu_1)C_{1r}(i)  \IEEEyesnumber\IEEEyessubnumber  \\
    \Lambda_2(i) =  (1-\mu_2)C_{2r}(i) \IEEEyessubnumber  \\
    \Lambda_3(i) = (1-\mu_1)C_{12r}(i)+(1-\mu_2)C_{21r}(i)\IEEEyessubnumber \\
    \Lambda_4(i) = \mu_2 C_{r1}(i) \IEEEyessubnumber\\
    \Lambda_5(i) = \mu_1 C_{r2}(i)\IEEEyessubnumber\\
    \Lambda_6(i) = \mu_1 C_{r2}(i)+\mu_2 C_{r1}(i)\IEEEyessubnumber
\end{IEEEeqnarray}
and $\mathcal{I}_k(i)\in\{0,1\}$ is a binary indicator variable which determines whether $\mathcal{M}_k$ is a \textit{possible candidate} for selection in the $i$-th time slot, i.e., mode $\mathcal{M}_k$ cannot be selected if $\mathcal{I}_k(i)=0$, and $\mu_1$ and $\mu_2$ are decision thresholds. For the values of $\mathcal{I}_k(i)$,  the optimal values of $\mu_1$ and $\mu_2$ in (\ref{SelecMet}), and the optimal value of $t(i)$ in  $C_{12r}(i)$ and $C_{21r}(i)$ in (\ref{SelecMet}c), we distinguish  three mutually exclusive SNR regions:
\vspace{0.1cm}

\noindent
\textbf{SNR Region $\mathcal{R}_1$:} SNR region $\mathcal{R}_1$ is chractarized by $\left\{ P_2>P_r  \,\,  \mathrm{AND}  \,\, \omega_1<1\right\}\,\,\mathrm{OR} \,\, \left\{ P_2<P_r  \,\,  \mathrm{AND}  \,\, \omega_2<1\right\}\,\,\mathrm{OR} \\ \left\{P_2\Equal P_r  \,  \mathrm{AND}  \,\omega_3^{l} \Less \omega_3^{u} \right\}$ with $\omega_1\Equal \frac{E\{qC_{r2}\}}{E\left\{(1-q)[C_r-C_{2r}]\right\}}, \omega_2=\frac{E\{qC_{2r}\}}{E\left\{(1-q)C_{r1}\right\}}, \omega_3^{l} = \frac{E\{C_{r2}\}}{E\{C_{r}\}}$, and $\omega_3^{u}=\frac{E\{C_{r1}\}}{E\{C_{r1}+C_{2r}\}}$, where $q(i)$ is defined as
\begin{IEEEeqnarray}{rl}\label{Q1}
   q(i)=\begin{cases}
1, &\mathrm{if}\,\,\,\,\, \left\{ P_2>P_r \,\, \mathrm{AND}\,\,  \Lambda_3 \leq \Lambda_6\big|_{\mu_1=1,\mu_2=\mu_2^*}^{t(i)=0,\,\,\forall i} \right\}\\  &\mathrm{OR} \left\{ P_2<P_r \,\, \mathrm{AND}\,\,  \Lambda_3 \geq \Lambda_6\big|_{\mu_1=\mu_1^*,\mu_2=0}^{t(i)=0,\,\,\forall i} \right\}
\\0,  & \mathrm{otherwise}
\end{cases}. \IEEEeqnarraynumspace
\end{IEEEeqnarray}
In (\ref{Q1}), if $P_2>P_r$,   $\mu_2^*$ is chosen such that $E\{(1-q)C_{2r}\}=E\{qC_{r1}\}$ holds, whereas  if $P_2<P_r$, $\mu_1^*$ is chosen such that $E\{q[C_r-C_{2r}]\}=E\{(1-q)C_{r2}\}$ holds. 

For this SNR region,   $t(i)=0$, $\forall i$,  and $\mathcal{I}_k(i)=0$, for $k=1,4$,  $\forall i$, whereas  $\mathcal{I}_k(i)$  for $k=2,3,5,6$, and the thresholds $\mu_1$ and $\mu_2$ are as follows
\begin{IEEEeqnarray}{llll}
\mathrm{If}\,P_2\Great P_r\,\mathrm{then}\,
{\begin{cases}
\mu_1\Equal 1 \\
\mu_2\Equal \mu_2^*
\end{cases}}
\mathrm{and} \,
{\begin{cases}
\mathcal{I}_5(i)\Equal 0 \\
\mathcal{I}_3(i)\Equal 1\Minus \mathcal{I}_2(i)\Equal X_1(i) \\
\mathcal{I}_6(i)\Equal 1
\end{cases}}\nonumber \\
\mathrm{If}\,P_2\Less P_r\,\mathrm{then}\,
{\begin{cases}
\mu_1\Equal \mu_1^* \\
\mu_2\Equal 0
\end{cases}} 
\mathrm{and} \,
{\begin{cases}
\mathcal{I}_2(i)\Equal 0,  \\
\mathcal{I}_6(i)\Equal 1\Minus \mathcal{I}_5(i)\Equal X_2(i)  \\
\mathcal{I}_3(i)\Equal 1
\end{cases}}\nonumber\\
\mathrm{If}\,P_2\Equal P_r\,\mathrm{then}\,
{\begin{cases}
\mu_1\Equal 1 \\
\mu_2\Equal 0
\end{cases}}
\mathrm{and} \,
{\begin{cases}
\mathcal{I}_2(i)\Equal  X_3(i) [1\Minus X_4(i)]\\
\mathcal{I}_3(i)\Equal X_3(i) X_4(i) \\
\mathcal{I}_5(i)\Equal [1\Minus X_3(i)][1\Minus X_5(i)]\\
\mathcal{I}_6(i)\Equal [1\Minus X_3(i)]X_5(i)
\end{cases}} \nonumber
\end{IEEEeqnarray}
where the coin flip probabilities  are $p_1=\omega_1, p_2=\omega_2,  p_4=\frac{(1-p_3)\omega_3^{l}}{p_3(1-\omega_3^{l})}, p_5=\frac{p_3(1-\omega_3^{u})}{(1-p_3)\omega_3^{u}}$, and for $p_3$, we have the freedom to choose any value in the interval $[\omega_3^{l}\,,\, \omega_3^{u}]$. 
\vspace{0.1cm}

\noindent
\textbf{SNR Region $\mathcal{R}_2$:} SNR region $\mathcal{R}_2$ is charactarized by $\left\{ P_1>P_r  \,\,  \mathrm{AND}  \,\, \omega_1<1\right\} \,\,  \mathrm{OR}  \,\, \left\{ P_1<P_r  \,\,  \mathrm{AND}  \,\, \omega_2<1\right\} \,\,  \mathrm{OR}  \\ \{P_1\Equal P_r \,  \mathrm{AND}  \, \omega_3^{l} \Less \omega_3^{u} \}$ with $\omega_1\Equal \frac{E\{qC_{r1}\}}{E\left\{(1-q)[C_r-C_{1r}]\right\}}, \omega_2\hspace{-0.8mm}=\hspace{-0.8mm}\frac{E\{qC_{1r}\}}{E\left\{(1-q)C_{r2}\right\}}, \omega_3^{l} = \frac{E\{C_{r1}\}}{E\{C_{r}\}}$, and $\omega_3^{u}=\frac{E\{C_{r2}\}}{E\{C_{r2}+C_{1r}\}}$, where $q(i)$ is defined as
\begin{IEEEeqnarray}{rl}\label{Q2}
   q(i)=\begin{cases}
1, &\mathrm{if}\,\,\,\,\, \left\{ P_1>P_r \,\, \mathrm{AND}\,\, \big[\Lambda_3 \leq \Lambda_6\big]_{\mu_1=\mu_1^*,\mu_2=1}^{t(i)=1,\,\,\forall i} \right\}\\  
&\mathrm{OR} \left\{ P_1<P_r \,\, \mathrm{AND}\,\, \big[\Lambda_3 \geq \Lambda_6\big]_{\mu_1=0,\mu_2=\mu_2^*}^{t(i)=1,\,\,\forall i}\right\}
\\0,  & \mathrm{otherwise}
\end{cases}. \IEEEeqnarraynumspace
\end{IEEEeqnarray}
In (\ref{Q2}), if $P_1 \Great P_r$, the threshold $\mu_1^*$ is chosen such that $E\{(1\Minus q)C_{1r}\}\Equal E\{qC_{r2}\}$ holds, whereas if $P_1\Less P_r$, the threshold $\mu_2^*$ is chosen such that $E\{q[C_r\Less C_{1r}]\}\Equal E\{(1\Less q)C_{r1}\}$ holds. 

For this SNR region,   $t(i)=1$, $\forall i$,  and $\mathcal{I}_k(i)=0$, for $k=2,5$,  $\forall i$, whereas $\mathcal{I}_k(i)$ for $k=1,3,4,6$, and the thresholds $\mu_1$ and $\mu_2$ are as follows
\begin{IEEEeqnarray}{llll}
\mathrm{If}\,P_1\Great P_r\,\mathrm{then}\,
{\begin{cases}
\mu_1\Equal \mu_1^* \\
\mu_2\Equal 1
\end{cases}}
\mathrm{and} \,
{\begin{cases}
\mathcal{I}_4(i)\Equal 0 \\
\mathcal{I}_3(i)\Equal 1\Minus \mathcal{I}_1(i)\Equal X_1(i) \\
\mathcal{I}_6(i)\Equal 1
\end{cases}}\nonumber\\
\mathrm{If}\,P_1\Less P_r\,\mathrm{then}\,
{\begin{cases}
\mu_1\Equal 0 \\
\mu_2\Equal \mu_2^*
\end{cases}}
\mathrm{and} \,
{\begin{cases}
\mathcal{I}_1(i)\Equal 0,  \\
\mathcal{I}_6(i)\Equal 1\Minus \mathcal{I}_4(i)\Equal X_2(i)  \\
\mathcal{I}_3(i)\Equal 1
\end{cases}}\nonumber\\
\mathrm{If}\,P_1\Equal P_r\,\mathrm{then}\,
{\begin{cases}
\mu_1\Equal 0 \\
\mu_2\Equal 1
\end{cases}}
\mathrm{and} \,
{\begin{cases}
\mathcal{I}_1(i)\Equal  X_3(i) [1\Minus X_4(i)]\\
\mathcal{I}_3(i)\Equal X_3(i) X_4(i) \\
\mathcal{I}_4(i)\Equal [1\Minus X_3(i)][1\Minus X_5(i)]\\
\mathcal{I}_6(i)\Equal [1\Minus X_3(i)]X_5(i)
\end{cases}} \nonumber
\end{IEEEeqnarray}
where the coin flip probabilities  are $p_1=\omega_1, p_2=\omega_2,  p_4=\frac{(1-p_3)\omega_3^{l}}{p_3(1-\omega_3^{l})}, p_5=\frac{p_3(1-\omega_3^{u})}{(1-p_3)\omega_3^{u}}$, and for $p_3$, we have the freedom to choose any value in the interval $[\omega_3^{l}\,,\, \omega_3^{u}]$. 
\vspace{0.1cm}

\noindent
\textbf{SNR Region $\mathcal{R}_0$:} The SNR region is $\mathcal{R}_0$ if and only if it is not $\mathcal{R}_1$ or $\mathcal{R}_2$. For SNR region $\mathcal{R}_0$, $\mathcal{I}_k(i)=0$ for $k=1,2,4,5$ and $\mathcal{I}_k(i)=1$ for $k=3,6$, $\forall i$, and $t(i)$ is given by
\begin{IEEEeqnarray}{C}\label{flip}
    t(i) = \begin{cases}
 0,  \quad \mathrm{if}  \,\,  t \leq 0\\
t,  \quad \mathrm{if}  \,\,  0 < t < 1\\
 1, \quad \mathrm{if}  \,\,  t \geq 1 
\end{cases} \IEEEyesnumber 
\end{IEEEeqnarray}
where $t=\frac{E\{q(C_r-C_{2r})\}-E\{(1-q)C_{r2}\}}{E\left\{q (C_r-C_{1r}-C_{2r})\right\}}$ and $q(i)$ is given by  
\begin{IEEEeqnarray}{C}\label{Q0} 
    q(i) = \begin{cases}
 1, \quad \mathrm{if}  \,\,  \frac{C_r(i)}{C_{r1}(i)+C_{r2}(i)} \geq \frac{\mu^*}{1-\mu^*} \\
 0, \quad \mathrm{if}  \,\, \mathrm{otherwise}
\end{cases} 
\end{IEEEeqnarray}
where $\mu^*$ is chosen such that $E\{qC_r\}\hspace{-1mm}=\hspace{-1mm}E\{\hspace{-0.8mm}(1\hspace{-0.8mm}-\hspace{-0.5mm}q)(C_{r1}\hspace{-1mm}+\hspace{-0.5mm}C_{r2})\hspace{-0.8mm}\}$ holds. Moreover, if $0< t < 1$, we obtain $\mu_1=\mu_2=\mu^*$. Otherwise, if $t \leq 0$ or $t \geq 1$, we obtain $\mu_1 \geq \mu_2$ and $\mu_1 \leq \mu_2$, respectively, and $\mu_1$ and $\mu_2$ are chosen such that $\mathrm{C1}$ and $\mathrm{C2}$ in (\ref{FixProb}) hold.
\end{theorem}
%
\begin{proof}
Please refer to Appendix \ref{AppKKT}. 
\end{proof}

We note that in the proposed optimal policy, regardless of the  SNR region, the multiple access and broadcast modes are always selected with non-zero probability. Only in the extreme case, when the average SNR of one of the links is much larger than the average SNR of the other link, some of the point-to-point modes are also selected in the optimal policy.
Also, the mode selection metric $\Lambda_k(i)$ includes two parts. The first part is the capacity of mode $\mathcal{M}_k$, and the second part are the decision thresholds $\mu_1$ and/or $\mu_2$. The capacity  is a function of the instantaneous CSI, but the decision thresholds $\mu_1$ and $\mu_2$ depend only on the statistics of the channels. Hence, the decision thresholds, $\mu_1$ and $\mu_2$, can be obtained offline
and  used as long as the channel statistics remain unchanged.  

\begin{remark}
To find the optimal values of $\mu_1$ and $\mu_2$, we need  one/two-dimensional searches over $[0\,,\,1]$. Specifically, in  SNR regions $\mathcal{R}_1$ and $\mathcal{R}_2$, at least one of the thresholds $\mu_1$ and $\mu_2$  is given. Therefore, at most one one-dimensional search is required. In SNR region $\mathcal{R}_0$, if $0\leq \omega \leq 1$, we obtain $\mu_1=\mu_2=\mu^*$, therefore, a one-dimensional search is required to determine $\mu^*$. However, if $\omega<0$ or $\omega>1$, we need a two-dimensional search to find $\mu_1$ and $\mu_2$.
\end{remark}

\section{Simulation Results}\label{SimRes}

In this section, we evaluate the maximum average sum rate in the considered bidirectional relay network for  Rayleigh fading. Thus, channel gains $S_1(i)$ and $S_2(i)$ follow exponential distributions with averages $\Omega_1$ and $\Omega_2$, respectively. Furthermore, the transmit powers of user 1 and user 2  are identical and set to  $P_1=P_2=10$ dB, and $\Omega_2=0$ dB. The number of time slots used in the simulations is   $N=10^6$. 

Fig.~\ref{Comparison} illustrates the maximum achievable sum rate versus $\Omega_1$. In order to simulate all  possible cases introduced in Theorem~\ref{FixProt}, we consider $P_r=5,10,15 \,\text{dB}$, which yields the solid lines in Fig.~\ref{Comparison}. As the value of $\Omega_1$ varies from $-10$ dB to $10$ dB in Fig.~\ref{Comparison}, the SNR region changes from $\mathcal{R}_2$ to $\mathcal{R}_0$ and then from $\mathcal{R}_0$ to $\mathcal{R}_1$. Since in SNR region $\mathcal{R}_2$  the bottleneck is  the relay-to-user 1 link, reducing $\Omega_1$ reduces the sum rate. On the other hand, in SNR region $\mathcal{R}_1$, the bottleneck link is the  relay-to-user 2 link, and therefore, for large $\Omega_1$, the  maximum average sum rate  saturates.
\begin{figure}
\centering
\psfrag{Y}[c][c][0.75]{$\text{Sum Rate}$}
\psfrag{X}[c][c][0.75]{$\Omega_1\,\text{(dB)}$}
\psfrag{L1}[l][c][0.45]{$\text{Proposed Selection Policy}\,(P_r=15\,\text{dB})$}
\psfrag{L2}[l][c][0.45]{$\text{Proposed Selection Policy}\,(P_r=10\,\text{dB})$}
\psfrag{L3}[l][c][0.45]{$\text{Proposed Selection Policy}\,(P_r=5\,\text{dB})$}
\psfrag{L4}[l][c][0.45]{$\text{Selection Policy in [8]}\,(P_r=10\,\text{dB})$}
\psfrag{L5}[l][c][0.45]{$\text{MABC}\,(P_r=10\,\text{dB})$}
\psfrag{L6}[l][c][0.45]{$\text{TDBC}\,(P_r=10\,\text{dB})$}
\psfrag{L7}[l][c][0.45]{$\text{Traditional Two-Way}\,(P_r=10\,\text{dB})$}
\psfrag{00}[c][c][0.65]{$-10$}
\psfrag{01}[c][c][0.65]{$-8$}
\psfrag{02}[c][c][0.65]{$-6$}
\psfrag{03}[c][c][0.65]{$-4$}
\psfrag{04}[c][c][0.65]{$-2$}
\psfrag{05}[c][c][0.65]{$0$}
\psfrag{06}[c][c][0.65]{$2$}
\psfrag{07}[c][c][0.65]{$4$}
\psfrag{08}[c][c][0.65]{$6$}
\psfrag{09}[c][c][0.65]{$8$}
\psfrag{10}[c][c][0.65]{$10$}
\psfrag{99}[c][c][0.65]{$0$}
\psfrag{12}[c][c][0.65]{$0.5$}
\psfrag{13}[c][c][0.65]{$1$}
\psfrag{14}[c][c][0.65]{$1.5$}
\psfrag{15}[c][c][0.65]{$2$}
\psfrag{16}[c][c][0.65]{$2.5$}
\psfrag{17}[c][c][0.65]{$3$}
\psfrag{18}[c][c][0.65]{$3.5$}
\psfrag{19}[c][c][0.65]{$4$}
\includegraphics[width=3.5 in]{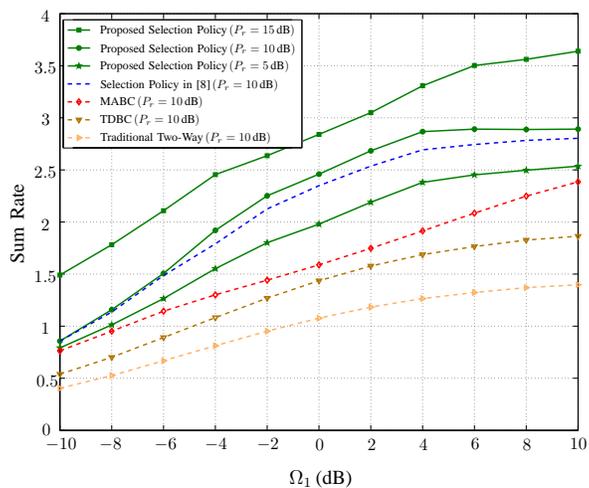}
\caption{Maximum achievable sum rate versus $\Omega_1$.}
\label{Comparison}
\end{figure}

As   performance benchmarks, we consider the traditional two-way, the TDBC \cite{TDBC},  the MABC \cite{MABC}, and the buffer-aided protocol in \cite{PopovskiLetter}. Since in \cite{PopovskiLetter}, the protocol  is derived for $P_1=P_2=P_r$,  in Fig.~3,  we only show the sum  rates of the benchmark schemes for $P_r=10 \,\text{dB}$. From the comparison in Fig.~3, we observe that a considerable gain is obtained by adaptive mode selection compared to the traditional two-way, TDBC, and MABC protocols. Furthermore, a small gain is obtained compared to the buffer-aided protocol in \cite{PopovskiLetter} which selects only the user 1-to-relay, user 2-to-relay, and broadcast modes adaptively. Finally, we note that our protocol is derived for  given $\Omega_1$, $\Omega_2$, $P_1$, $P_2$, and $P_r$,  and without a constraint on the total average power consumed by all nodes. The problem with a total average power constraint would lead to a different solution.

\section{Conclusion}\label{Conclusion}
We have derived the optimal transmission mode selection policy for sum rate maximization of the three-node half-duplex bidirectional buffer-aided relay network with fading links and fixed transmit powers at the nodes. The proposed selection policy determines the optimal transmission mode based on the instantaneous CSI of the involved links in each time slot and their long-term statistics. Simulation results confirmed that the proposed selection policy outperforms the existing protocols from the literature in terms of average sum rate.

\appendices


\section{Proof of Theorem 1 (Optimal Queue Condition)}
\label{AppQueue}

Let $\bar{C}_{r1}$ and $\bar{C}_{r2}$ denote the average rates achieved by the relay assuming that buffers $B_1$ and $B_2$ always have enough information to supply. Then, $\bar{C}_{r1}$ and $\bar{C}_{r2}$ are given by
\begin{IEEEeqnarray}{lll} \label{Capacity}  
		\bar{C}_{r1}= \underset{N\to \infty}{\lim}\frac{1}{N}\mathop \sum \limits_{i = 1}^N \left[ q_4(i)+q_6(i)\right]C_{r1}(i) \IEEEyesnumber \IEEEyessubnumber\\
\bar{C}_{r2}= \underset{N\to \infty}{\lim}\frac{1}{N}\mathop \sum \limits_{i = 1}^N \left[ q_5(i)+q_6(i)\right]C_{r2}(i). \IEEEyessubnumber
\end{IEEEeqnarray}
If we assume that $\bar{R}_{2r}\geq\bar{C}_{r1}$ holds, i.e., the average rate flowing into the buffer is larger than or equal to the average capacity that can  flow out of the buffer, then the buffer always has enough information to supply and $\bar{R}_{r1}$ in (\ref{RatReg123}) becomes $\bar{R}_{r1}=\bar{C}_{r1}$, where $\bar{C}_{r1}$ is given in (\ref{Capacity}a).  For a proof of this property, we refer to the proof of Theorem 1 in \cite[Appendix A]{NikolaJSAC}.  On the other hand, if we assume that $\bar{R}_{2r}<\bar{C}_{r1}$ holds, i.e., the average information flowing into the buffer is less than the average capacity that can flow  out of the buffer, then by the law of conservation of flow, we obtain $\bar{R}_{r1}=\bar{R}_{2r}$. Similar  results hold for $\bar{R}_{r2}$, $\bar{R}_{1r}$, and $\bar{C}_{r2}$. Hence, $\bar{R}_{r1}$ and $\bar{R}_{r2}$ can be written as
\begin{IEEEeqnarray}{lll}\label{eq_1}
   \bar{R}_{r1} = \begin{cases}
\bar{C}_{r1}, \,\, \mathrm{if}\,\, \bar{R}_{2r}\geq\bar{C}_{r1} \\
\bar{R}_{2r}, \,\, \mathrm{otherwise}
\end{cases} \IEEEyesnumber\IEEEyessubnumber \\
\bar{R}_{r2} = \begin{cases}
\bar{C}_{r2}, \,\, \mathrm{if}\,\, \bar{R}_{1r}\geq\bar{C}_{r2} \\
\bar{R}_{1r}, \,\, \mathrm{otherwise}
\end{cases} \IEEEyessubnumber
\end{IEEEeqnarray}
On the other hand, (\ref{eq_1}) can also be written as 
\begin{IEEEeqnarray}{lll}\label{eq_2}
\bar{R}_{r1} = \min\{\bar{R}_{2r},\bar{C}_{r1}\}\label{eq_2b}\IEEEyesnumber\IEEEyessubnumber \\
    \bar{R}_{r2} = \min\{\bar{R}_{1r},\bar{C}_{r2}\}\IEEEyessubnumber
\end{IEEEeqnarray}
i.e.,  the average   rate flowing out of buffer $B_1$ ($B_2$) is limited by the minimum of the average rate flowing into the buffer  and the capacity rate that can flow out of the buffer.
  We will show that, in the optimal selection policy,   $\bar{R}_{1r}=\bar{C}_{r2}$ and    $\bar{R}_{2r}=\bar{C}_{r1}$  must hold, and as a result from (\ref{eq_2}a) and (\ref{eq_2}b), we obtain  $\bar{R}_{1r}=\bar{R}_{r2}=\bar{C}_{r2}$ and $\bar{R}_{2r}=\bar{R}_{r1}=\bar{C}_{r1}$.  
 To show this, we denote by $I_k$ the set comprising the time slots $i$ in which $q_k(i)=1$, i.e., the elements in $I_k$ represent the time slots in which transmission mode $\mathcal{M}_k$ is selected in the optimal solution. Note that $\sum_{k=1}^6 |I_k|=N$, where $|\cdot|$ denotes the cardinality of a set and $\varoslash$ denotes an empty set with zero cardinality. We will show that, except for the case when both arguments of the $\min$ functions in (\ref{eq_2}a) and (\ref{eq_2}b)  are equal, all other cases result in suboptimal sum rate or can be safely replaced with the case when  both arguments of the $\min$ function are equal without reducing the sum rate. For the two $\min$ functions,  we have to consider the following nine cases:\\*
\textbf{Case 1:} If $\bar{R}_{1r}\hspace{-1mm}>\hspace{-1mm}\bar{C}_{r2}$ and $\bar{R}_{2r}\hspace{-1mm}>\hspace{-1mm}\bar{C}_{r1}$, then we can move any indices $i$ from   $I_1, I_2$, or $I_3$ to $I_6$ and thereby increase both $\bar{C}_{r1}$ and $\bar{C}_{r2}$, and thus improve the sum rate. Hence, $\bar{R}_{1r}\hspace{-1mm}>\hspace{-1mm}\bar{C}_{r2}$ and $\bar{R}_{2r}\hspace{-1mm}>\hspace{-1mm}\bar{C}_{r1}$ cannot hold since the resulting sum rate is suboptimal.\\*
\textbf{Case 2:} If $\bar{R}_{1r}\hspace{-1mm}=\hspace{-1mm}\bar{C}_{r2}$ and $\bar{R}_{2r}\hspace{-1mm}>\hspace{-1mm}\bar{C}_{r1}$, then if $I_2 \neq \varoslash$, one can move some indices from $I_2$ to $I_4$  which maintains $\bar{R}_{1r}\hspace{-1mm}=\hspace{-1mm}\bar{C}_{r2}$ and increases $\bar{C}_{r1}$, and consequently increases the sum rate, thereby contradicting optimality. On the other hand, if $I_2 = \varoslash$, we have $I_3 \neq \varoslash $ since otherwise, $\bar{R}_{2r}=0$ and because we have assumed $\bar{R}_{2r}\hspace{-1mm}>\hspace{-1mm}\bar{C}_{r1}$, we obtain $\bar{R}_{2r}\hspace{-1mm}=\hspace{-1mm}\bar{C}_{r1}=0$. Hence, assume $I_2 = \varoslash$ and $I_3 \neq \varoslash$, then there are two possibilities. The first is when   $t(i) = 1$, $ \forall i$  for which $q_3(i)=1$. Then, since $C_{12r}(i)$ becomes $C_{12r}(i) = C_{1r}(i)$, we can move some indices from $I_3$ to $I_1$, which maintains $\bar{R}_{1r}\hspace{-1mm}=\hspace{-1mm}\bar{C}_{r2}$ and decreases $\bar{R}_{2r}$ until we obtain $\bar{R}_{2r}\hspace{-1mm}=\hspace{-1mm}\bar{C}_{r1}$, thereby contradicting $\bar{R}_{2r}\hspace{-1mm}>\hspace{-1mm}\bar{C}_{r1}$. However, if $t(i)\neq 1$ for some time slots $i$ for which $q_3(i)=1$, then since $C_{12r}(i)\leq C_{1r}(i)$ holds, we can move some indices from $I_3$ to $I_1$ and $I_4$ such that $\bar{R}_{1r}\hspace{-1mm}=\hspace{-1mm}\bar{C}_{r2}$ is maintained and $\bar{C}_{r1}$ is increased and thereby the sum rate is increased, which results in a contradiction of the optimal sum rate assumption. \\*
\textbf{Case 3:} If $\bar{R}_{1r}\hspace{-1mm}=\hspace{-1mm}\bar{C}_{r2}$ and $\bar{R}_{2r}\hspace{-1mm}<\hspace{-1mm}\bar{C}_{r1}$, then if $I_4 \neq \varoslash$,  we can move some indices of $I_4$ to $I_2$ which maintains $\bar{R}_{1r}\hspace{-1mm}=\hspace{-1mm}\bar{C}_{r2}$ and increases $\bar{R}_{2r}$ and consequently increases the sum rate and thus contradicts optimality. On the other hand, if $I_4 = \varoslash$,  we must have $\bar{I}_6 \neq \varoslash$, otherwise, $\bar{C}_{r2}=0$, and since we have assumed $\bar{R}_{2r}\hspace{-1mm}<\hspace{-1mm}\bar{C}_{r1}$, we obtain $\bar{R}_{2r}\hspace{-1mm}=\hspace{-1mm}\bar{C}_{r1}=0$. However, if $I_6 \neq \varoslash$, we can move some of the indices in $I_6$ to $I_5$ which maintains $\bar{R}_{1r}\hspace{-1mm}=\hspace{-1mm}\bar{C}_{r2}$ and decreases $\bar{C}_{r1}$ until we obtain $\bar{R}_{2r}\hspace{-1mm}=\hspace{-1mm}\bar{C}_{r1}$, thereby contradicting $\bar{R}_{2r}\hspace{-1mm}<\hspace{-1mm}\bar{C}_{r1}$.\\*
\textbf{Case 4:} If $\bar{R}_{1r}\hspace{-1mm}>\hspace{-1mm}\bar{C}_{r2}$ and $\bar{R}_{2r}\hspace{-1mm}<\hspace{-1mm}\bar{C}_{r1}$, then if $I_4 \neq \varoslash$,  we can move some indices of $I_4$ to $I_5$, which increases $\bar{C}_{r2}$ and consequently the sum rate and thereby contradicting optimality. However, if $I_4 = \varoslash$, then we must have   $I_6 \neq \varoslash$  otherwise  $\bar{C}_{r1}=0$ leading to $\bar{R}_{2r}\hspace{-1mm}=\hspace{-1mm}\bar{C}_{r1}=0$ and the rest of the conclusions are similar to that in Case 2. Then, if $I_6 \neq \varoslash$, one can move some of the indices from $I_6$ to $I_5$ which maintains $\bar{R}_{1r}\hspace{-1mm}>\hspace{-1mm}\bar{C}_{r2}$ but decreases $\bar{C}_{r1}$ until we obtain $\bar{R}_{2r}\hspace{-1mm}=\hspace{-1mm}\bar{C}_{r1}$ and the rest of conclusions are similar to those in Case 2.

The conclusions for the four remaining cases, ($\bar{R}_{1r}\hspace{-1mm}<\hspace{-1mm}\bar{C}_{r2}$ and $\bar{R}_{2r}\hspace{-1mm}<\hspace{-1mm}\bar{C}_{r1}$), ($\bar{R}_{1r}\hspace{-1mm}>\hspace{-1mm}\bar{C}_{r2}$ and $\bar{R}_{2r}\hspace{-1mm}=\hspace{-1mm}\bar{C}_{r1}$), ($\bar{R}_{1r}\hspace{-1mm}<\hspace{-1mm}\bar{C}_{r2}$ and $\bar{R}_{2r}\hspace{-1mm}=\hspace{-1mm}\bar{C}_{r1}$), and ($\bar{R}_{1r}\hspace{-1mm}<\hspace{-1mm}\bar{C}_{r2}$ and $\bar{R}_{2r}\hspace{-1mm}>\hspace{-1mm}\bar{C}_{r1}$) are similar to those in Case 1, Case 2, Case 3 and Case 4, respectively. Therefore, the final case remaining for the optimal solution is $\bar{R}_{r2} \hspace{-1mm} = \hspace{-1mm}\bar{R}_{1r}\hspace{-1mm}=\hspace{-1mm}\bar{C}_{r2}$ and $\bar{R}_{r1}\hspace{-1mm} =\hspace{-1mm} \bar{R}_{2r}\hspace{-1mm}=\hspace{-1mm}\bar{C}_{r1}$. This completes the proof.


\section{Proof of Theorem 2 (Mode Selection Protocol)}
\label{AppKKT}

In this appendix, we solve the optimization problem given by (\ref{FixProb}). We first relax the binary condition for $q_k(i)$, i.e., $q_k(i)[1-q_k(i)]=0$, to $0\leq q_k(i)\leq 1$, and later in Appendix \ref{AppBinRelax}, we prove  that the binary relaxation does not affect the maximum average sum rate. In the following, we investigate the Karush-Kuhn-Tucker (KKT) necessary conditions of the relaxed optimization problem and then show that the necessary conditions result in a unique sum rate and thus the solution is optimal. 

To simplify the usage of the KKT conditions, we formulated a minimization problem equivalent to the relaxed maximization problem in (\ref{FixProb}) as follows
\begin{IEEEeqnarray}{Cll}\label{FixProbMin}
    {\underset{q_k(i),t(i),\,\,\forall i,k}{\mathrm{minimize}}}\,\, &-(\bar{R}_{1r}+\bar{R}_{2r}) \nonumber \\
    \mathrm{subject\,\, to} \,\, &\mathrm{C1}:\quad \bar{R}_{1r}-\bar{R}_{r2}=0  \nonumber \\
    &\mathrm{C2}:\quad \bar{R}_{2r}-\bar{R}_{r1}=0 \nonumber \\
		&\mathrm{C3}:\quad \mathop \sum \nolimits_{k = 1}^6 {q_k}\left( i \right) - 1 =0, \,\, \forall i   \nonumber \\
    &\mathrm{C4}:\quad q_k(i)-1 \leq 0, \,\, \forall i, k \nonumber \\
&\mathrm{C5}:\quad -q_k(i) \leq 0, \,\, \forall i, k \nonumber \\
    &\mathrm{C6}:\quad t(i)-1 \leq 0, \,\, \forall i \nonumber \\
&\mathrm{C7}:\quad -t(i) \leq 0, \,\, \forall i. \IEEEyesnumber
\end{IEEEeqnarray}
The Lagrangian function for the above optimization problem is provided in (\ref{KKT Function}) at the top of the next page where $\mu_1,\mu_2,\lambda(i),\alpha_k(i),\beta_k(i),\phi_1(i)$ and $\phi_2(i)$ are the Lagrange multipliers corresponding to constraints $\mathrm{C1,C2,C3,C4,}\\ \mathrm{C5,C6}$ and $\mathrm{C7}$, respectively. The KKT conditions include the following:

\begin{figure*}[!t]
\normalsize
\begin{IEEEeqnarray}{l}\label{KKT Function}
   \underset{\mathrm{for}\,\, \forall i,k,l}{\mathcal{L}(q_k(i),t(i),\lambda(i),\alpha_k(i),\beta_k(i),\phi_l(i))}  =  \nonumber \\  \qquad- (\bar{R}_{1r}+\bar{R}_{2r}) + \mu_1(\bar{R}_{1r}-\bar{R}_{r2}) + \mu_2(\bar{R}_{2r}-\bar{R}_{r1}) 
    + \mathop \sum \limits_{i = 1}^N \lambda \left( i \right)\left( {\mathop \sum \limits_{k = 1}^6 {q_k}\left( i \right) - 1} \right) \nonumber \\ 
\qquad + \mathop \sum \limits_{i = 1}^N \mathop \sum \limits_{k = 1}^6 {\alpha _k}\left( i \right)\left( {{q_k}\left( i \right) - 1} \right) - \mathop \sum \limits_{i = 1}^N \mathop \sum \limits_{k = 1}^6 {\beta _k}\left( i \right){q_k}\left( i \right) + \mathop \sum \limits_{i = 1}^N \phi_1(i) (t(i)-1) - \mathop \sum \limits_{i = 1}^N \phi_0(i) t(i) \IEEEeqnarraynumspace \IEEEyesnumber
\end{IEEEeqnarray}
\hrulefill

\vspace*{4pt}
\end{figure*}

\noindent
\textbf{1)} Stationary condition: The differentiation of the Lagrangian function with respect to the primal variables, $q_k(i)$ and $t(i)$ for $\forall i,k$, is equal to zero for the optimal solution, i.e.,
\begin{IEEEeqnarray}{CCCl}\label{Stationary Condition}
    \frac{\partial\mathcal{L}}{\partial q_k(i)} &=& 0 \quad &\mathrm{for} \,\, \forall i,k \IEEEyesnumber\IEEEyessubnumber\\
\frac{\partial\mathcal{L}}{\partial t(i)} &=& 0 \quad &\mathrm{for} \,\, \forall i.\IEEEyessubnumber
\end{IEEEeqnarray}

\noindent
\textbf{2)} Primal feasibility condition: The optimal solution has to satisfy the constraints of the primal problem in (\ref{FixProbMin}).

\noindent
\textbf{3)} Dual feasibility condition: The Lagrange multipliers for the inequality constraints have to be non-negative, i.e.,
\begin{IEEEeqnarray}{lll}\label{Dual Feasibility Condition}
            \alpha_k(i)\geq 0,  &\forall i,k \IEEEyesnumber\IEEEyessubnumber\\
\beta_k(i)\geq 0, &\forall i,k \IEEEyessubnumber\\
             \phi_l(i) \geq0,   & \forall i,l \IEEEyessubnumber
\end{IEEEeqnarray}
have to hold.

\noindent
\textbf{4)} Complementary slackness: If an inequality is inactive, i.e., the optimal solution is in the interior of the corresponding set, the corresponding Lagrange multipliers are zero. Thus, we obtain
\begin{IEEEeqnarray}{lll}\label{Complementary Slackness}
    {\alpha _k}\left( i \right)\left( {{q_k}\left( i \right) - 1} \right)=0,\quad  &\forall i,k \IEEEyesnumber\IEEEyessubnumber\\
    {\beta _k}\left( i \right){q_k}\left( i \right)=0,  &\forall i,k \IEEEyessubnumber\\
    \phi_1(i) (t(i)-1) = 0, &\forall i \IEEEyessubnumber\\
		\phi_0(i) t(i) = 0, &\forall i.\IEEEyessubnumber
\end{IEEEeqnarray}
A common approach to find a set of primal variables, i.e., $q_k(i), t(i),\,\,\forall i,k$ and Lagrange multipliers, i.e., $\mu_1,\mu_2,\lambda(i),\\ \alpha_k(i),\beta_k(i),\phi_l(i),\,\,
\forall i,k,l$, which satisfy the KKT conditions is to start with the complementarity slackness conditions and see if the inequalities are active or not. Combining these results with the primal feasibility and dual feasibility conditions, we obtain various possibilities. Then, from these possibilities, we obtain one or more candidate solutions from the stationary conditions and the optimal solution is surely one of these candidates. In the following subsections, we first pursue this procedure in order to find optimal $q_k^*(i)$ and $t^*(i)$ for $\forall i,k$, and then, we construct the protocol introduced in Theorem \ref{FixProt} based on three exclusive SNR regions.

\subsection{Optimal $q_k^*(i)$}

In order to determine the optimal selection policy, we differentiate the Lagrangian function in (\ref{KKT Function}) with respect to   $q_k(i)$, for $k=1,...,6$, equate the result to zero, and obtain the corresponding $q_k(i)$. This leads to
\begin{IEEEeqnarray}{lll}\label{Stationary Mode}
    q_1^*(i): \,\, &-\frac{1}{N}(1-\mu_1)C_{1r}(i)+\lambda(i)+\alpha_1(i)-\beta_1(i)=0\,\, \IEEEeqnarraynumspace \IEEEyesnumber \IEEEyessubnumber\\
    q_2^*(i): \,\, &-\frac{1}{N}(1-\mu_2)C_{2r}(i)+\lambda(i)+\alpha_2(i)-\beta_2(i)=0 \IEEEyessubnumber \\
    q_3^*(i): \,\, &-\frac{1}{N}[(1-\mu_1)C_{12r}(i)+(1-\mu_2)C_{21r}(i)] +\lambda(i)+\alpha_3(i)-\beta_3(i)=0 \IEEEeqnarraynumspace \IEEEyessubnumber \\
    q_4^*(i): \,\, &-\frac{1}{N}\mu_2 C_{r1}(i)+\lambda(i)+\alpha_4(i)-\beta_4(i)=0 \IEEEyessubnumber\\
    q_5^*(i): \,\, &-\frac{1}{N}\mu_1 C_{r2}(i)+\lambda(i)+\alpha_5(i)-\beta_5(i))=0 \IEEEyessubnumber\\
    q_6^*(i): \,\, &-\frac{1}{N}[\mu_1 C_{r2}(i)+\mu_2 C_{r1}(i)]+\lambda(i) +\alpha_6(i)-\beta_6(i)=0. \IEEEeqnarraynumspace\IEEEyessubnumber
\end{IEEEeqnarray}
Without loss of generality, we obtain the necessary condition for $q_1^*(i)=1$ and then generalize the result to $q_k^*(i)=1$ for $k=2,\dots,6$. From constraint $\mathrm{C3}$ in (\ref{FixProbMin}), the other selection variables are zero, i.e., $q_k^*(i)=0$, for $k=2,3,4,5,6$. Furthermore, from (\ref{Complementary Slackness}a), we obtain $\alpha_k(i)=0$, for $k = 2,3,4,5,6$ and from (\ref{Complementary Slackness}b), we obtain $ \beta_1(i)=0$. Then, by substituting these values into (\ref{Stationary Mode}), we obtain
\begin{IEEEeqnarray}{lll}\label{MET}
    \lambda(i)+\alpha_1(i) = (1-\mu_1)C_{1r}(i) \triangleq \Lambda_1(i)  \IEEEyesnumber\IEEEyessubnumber  \\
    \lambda(i)-\beta_2(i) =  (1-\mu_2)C_{2r}(i) \triangleq \Lambda_2(i)\IEEEyessubnumber  \\
   \lambda(i)-\beta_3(i) =  (1\Minus \mu_1)C_{12r}(i)\Add (1\Minus \mu_2)C_{21r}(i) \hspace{-1mm} \triangleq \hspace{-1mm} \Lambda_3(i) \quad\,\,\,\, \IEEEyessubnumber \\
   \lambda(i)-\beta_4(i)  = \mu_2 C_{r1}(i) \triangleq \Lambda_4(i) \IEEEyessubnumber\\
    \lambda(i)-\beta_5(i) = \mu_1 C_{r2}(i) \triangleq \Lambda_5(i) \IEEEyessubnumber\\
    \lambda(i)-\beta_6(i) = \mu_1 C_{r2}(i)+\mu_2 C_{r1}(i) \triangleq \Lambda_6(i)\IEEEyessubnumber
\end{IEEEeqnarray}
where $\Lambda_k(i)$ is referred  to as selection metric. By subtracting (\ref{MET}a) from the rest of the equations in (\ref{MET}), we obtain
\begin{IEEEeqnarray}{rCl}\label{eq_2_1}
    \Lambda_1(i) - \Lambda_k(i) = \alpha_1(i)+\beta_k(i), \quad k=2,3,4,5,6. \IEEEyesnumber
\end{IEEEeqnarray}
From the dual feasibility conditions given in (\ref{Dual Feasibility Condition}a) and (\ref{Dual Feasibility Condition}b), we have $\alpha_k(i),\beta_k(i)\geq 0$. By inserting $\alpha_k(i),\beta_k(i)\geq 0$ in (\ref{eq_2_1}), we obtain the necessary condition for $q_1^*(i)=1$ as
\begin{IEEEeqnarray}{lll}
    \Lambda_1(i) \geq \max \left \{ \Lambda_2(i), \Lambda_3(i), \Lambda_4(i), \Lambda_5(i), \Lambda_6(i) \right \}. \IEEEyesnumber
\end{IEEEeqnarray}
Repeating the same procedure for $q_k^*(i)=1,\,\,k=2,\dots,6$, we obtain the necessary condition for selecting  transmission mode $\mathcal{M}_{k^*}$ in the $i$-th time slot as 
\begin{IEEEeqnarray}{lll}\label{OptMet}
   \Lambda_{k^*}(i) \geq {\underset{k\in\{1,\cdots,6\}}{\max}}\{\Lambda_{k}(i)\}, \IEEEyesnumber
\end{IEEEeqnarray}
where the Lagrange multipliers $\mu_1$ and $\mu_2$ are chosen such that $\mathrm{C1}$ and $\mathrm{C2}$ in (\ref{FixProbMin}) hold and the optimal value of $t(i)$ in $C_{12r}(i)$ and $C_{21r}(i)$ is obtained in the next subsection. 
We note that $\mu_1$ and $\mu_2$ are long-term thresholds (Lagrange multipliers) which only depend on the statistics of the channels. Moreover, we prove in Appendix \ref{AppMURegion} that the optimal values of $\mu_1$ and $\mu_2$ are in the interval $[0\,\,1]$. Moreover, for $\mu_1\neq 0,1$ and $\mu_2\neq 0,1$, the probability that two selection metrics in (\ref{MET}) are equal is zero due to the randomness of the time-continuous channel gains. Therefore, the necessary condition for selecting transmission mode $\mathcal{M}_k$ in (\ref{OptMet}) is in fact sufficient and is the optimal selection policy. However, if $\mu_1 = 0,1$ or $\mu_2= 0,1$, then some of the selection metrics in (\ref{MET}) are equal for all time slots. In this case, we have the freedom to choose between the modes which have the same value of the selection metric in (\ref{MET}) as long as the long-term constraints $\mathrm{C1}$ and $\mathrm{C2}$ in (\ref{FixProbMin}) hold. One way of selecting between the modes for which the selection metrics $\Lambda_{k}(i)$ are identical is    to use a probabilistic approach via coin flips.  Therefore, let $X_n(i)\in \{\mathrm{0,1}\}$ be the outcomes of the $n$-th coin flip in the $i$-th time slot with the probabilities of the possible outcomes defined as $\Pr\{X_n(i)=1\}=p_n$ and $\Pr\{X_n(i)=0\}=1-p_n$. In order to include the coin flip in (\ref{OptMet}), we write the necessary and sufficient condition for the selection of mode $\mathcal{M}_{k}$ in the $i$-th time slot as
\begin{IEEEeqnarray}{lll}\label{OptMet2}
 \Lambda_{k}(i) > {\underset{k\in\{1,\cdots,6\}}{\max}}\{\mathcal{I}_{k}(i)\Lambda_{k}(i)\} \IEEEyesnumber
\end{IEEEeqnarray}
where $\mathcal{I}_{k}(i)\in\{0,1\}$ is a binary indicator variable which is equal to zero for the modes that can not be selected in the $i$-th time slot and one otherwise. As will be seen, the value of $\mathcal{I}_{k}(i)$, depending on the average SNRs, is either deterministic or obtained  as a function of the outcome of the  coin flip(s).

\subsection{Optimal $t^*(i)$}
To obtain the optimal $t^*(i)$ for the multiple access mode, we assume $q^*_3(i)=1$ and calculate the derivative in (\ref{Stationary Condition}b) which leads to
\begin{IEEEeqnarray}{lll} \label{Stationary t}
    t^*(i):\,\, -\frac{1}{N}(\mu_1-\mu_2) [ C_r(i) - C_{1r}(i) - C_{2r}(i)] + \phi_1(i)-\phi_0(i)=0.  \IEEEeqnarraynumspace\IEEEyesnumber
\end{IEEEeqnarray}
Now, we investigate the following possible cases for $t^*(i)$:

\noindent
\textbf{Case 1:} If $0<t^*(i)<1$, then from (\ref{Complementary Slackness}c) and (\ref{Complementary Slackness}d), we obtain $\phi_l(i)=0,$ for $l=0,1$. From (\ref{Stationary t}) and by considering  $C_r(i) - C_{1r}(i) - C_{2r}(i)\leq 0$, we also obtain $\mu_1=\mu_2=\mu$. Combining these results into (\ref{OptMet}), leads to
\begin{IEEEeqnarray}{lll} \label{MU12}
    \Lambda_3(i) = (1-\mu)C_r(i) \geq \max \{\Lambda_1(i),\Lambda_2(i)\} \IEEEyesnumber\IEEEyessubnumber \\
\Lambda_6(i) = \mu C_{r1}(i) + \mu C_{r2}(i)  \geq \max \{\Lambda_4(i),\Lambda_5(i)\}. \IEEEeqnarraynumspace \IEEEyessubnumber
\end{IEEEeqnarray}
The inequalities in (\ref{MU12}a) and  (\ref{MU12}b) hold with equality with non-zero probability if and only if $\mu=1$ and $\mu=0$, respectively, however, this leads to a contradiction as shown in Appendix \ref{AppMURegion}. Therefore, only $\mathcal{M}_3$ and $\mathcal{M}_6$ are selected for the optimal strategy. Hence, we can use the notation in (\ref{OptMet2}), by  setting $\mathcal{I}_k(i)=0, \,\, k=1,2,4,5$ and $\mathcal{I}_k(i)=1,\,\,k=3,6$, $\forall i$.

In this case, the values of $t(i)$, $\forall i$, and the value of $\mu$ are found such that   constraints $\mathrm{C1}$ and $\mathrm{C2}$ in (\ref{FixProbMin})  hold. Note that when  $\mathrm{C1}$ and $\mathrm{C2}$ hold, the following also holds
\begin{IEEEeqnarray}{rCl}\label{EquCons}
    \bar{R}_{1r}+\bar{R}_{2r}&=&\bar{R}_{r1}+\bar{R}_{r2}. \IEEEyesnumber
\end{IEEEeqnarray}
Now,  since $\Lambda_3(i)$ and $\Lambda_6(i)$  are the only selection metrics for which $I_k(i)\neq 0$, for $k=3,6$, $\forall i$, we observe that both $\Lambda_3(i)$ and $\Lambda_6(i)$, and both the left and right hand side in (\ref{EquCons}) are independent of the value of $t(i)$, $\forall i$. Thus, we can set $\mu$ independently from $t(i)$ such that (\ref{EquCons}) is satisfied. We note that the optimal $\mu$ is unique since increasing $\mu$ increases the right hand side of (\ref{EquCons}) and decreases the left hand side, simultaneously. Since the sum rate does not depend on $t(i)$, i.e., $\bar{R}_{1r}+\bar{R}_{2r}=\frac{1}{N}\mathop \sum_{i = 1}^N q_3(i) C_{r}(i)$, any choice of $t(i)$, $\forall i$, that satisfies $\mathrm{C1}$ in (\ref{FixProbMin}) is optimal. Then, if $\mathrm{C1}$ in (\ref{FixProbMin}) and (\ref{EquCons}) hold, $\mathrm{C2}$ in (\ref{FixProbMin}) also holds. As a result, we have satisfied both $\mathrm{C1}$ and $\mathrm{C2}$ by satisfying $\mathrm{C1}$ in (\ref{FixProbMin}) and (\ref{EquCons}). Therefore, we assume $t(i)=t,\,\,\forall i$, and $t$ is given by
\begin{IEEEeqnarray}{rCl}\label{TimeVar}
   t=\frac{E\{q(C_r-C_{2r})\}-E\{(1-q)C_{r2}\}}{E\left\{q (C_r-C_{1r}-C_{2r})\right\}} \IEEEeqnarraynumspace \IEEEyesnumber
\end{IEEEeqnarray}
where 
\begin{IEEEeqnarray}{C} 
    q(i) = \begin{cases}
 1, \quad \mathrm{if}  \,\,  \frac{C_r(i)}{C_{r1}(i)+C_{r2}(i)} \geq \frac{\mu}{1-\mu} \\
 0, \quad  \mathrm{otherwise}
\end{cases}. 
\end{IEEEeqnarray}

We note that the necessary condition for the optimality of this case is that $0< t < 1$. For the cases which lead to $t\leq 0$ or $t\geq 1$, we cannot have $\mu_1=\mu_2$, and therefore, we obtain either $t(i)=1$ or $t(i)=0$. In these cases, we need different selection policies which are investigated in the following.

\noindent
\textbf{Case 2:} If $t(i)=0$, then from (\ref{Complementary Slackness}c), we obtain $\phi_1(i)=0$ and from (\ref{Dual Feasibility Condition}c), we obtain $\phi_0(i)\geq 0$. Combining these results with (\ref{Stationary t}), we obtain $\mu_1\geq\mu_2$. Then, by substituting these results in (\ref{MET}), we obtain
\begin{IEEEeqnarray}{rll}\label{M25}
    \Lambda_3(i) &=& (1-\mu_1)[C_{r}(i)-C_{2r}(i)]+(1-\mu_2)C_{2r}(i) \nonumber \\
     &=&(1-\mu_1)C_{r}(i)+ (\mu_1-\mu_2)C_{2r}(i) \nonumber \\ &\geq& \max \{\Lambda_1(i),\Lambda_2(i)\}  \IEEEyesnumber \IEEEyessubnumber \\
\Lambda_6(i) &=& \mu_2 C_{r1}(i) + \mu_1 C_{r2}(i)\geq \max \{\Lambda_4(i),\Lambda_5(i)\}. \IEEEeqnarraynumspace \IEEEyessubnumber
\end{IEEEeqnarray}
The expressions in (\ref{M25}a) and  (\ref{M25}b) hold with equality with non-zero probability if and only if $\mu_1=1$ and $\mu_2=0$, respectively, otherwise only the inequality holds, respectively. Therefore, we consider the following four possibilities for the thresholds $\mu_1$ and $\mu_2$.

\noindent
\textbf{a)} If $0\Less\mu_1 \Less 1$ and $0\Less\mu_2 \Less 1$, then, (\ref{M25}) holds with inequality and therefore, only $\mathcal{M}_3$ and $\mathcal{M}_6$ are selected. Hence, for this case, we can set $\mathcal{I}_k(i)=0,\,\, k=1,2,4,5$ and $\mathcal{I}_k(i)=1,\,\,k=3,6$ for $\forall i$. Since by simultaneously increasing $\mu_1$ and $\mu_2$,  the left hand side of (\ref{EquCons}) decreases and  the right hand side increases, we can always find a pair $(\mu_1,\mu_2)$ which satisfies this constraint. Therefore, the necessary condition for the optimality in this case is that among the candidates $(\mu_1,\mu_2)$ that satisfy (\ref{EquCons}), there exists a point which satisfies either $\mathrm{C1}$ or $\mathrm{C2}$ in (\ref{FixProbMin}) (we note that since we assume (\ref{EquCons}) holds, if $\mathrm{C1}$ holds, $\mathrm{C2}$ must hold too and vice versa). Moreover, if there exists such a point then it is unique. To prove this, note that the candidates $(\mu_1,\mu_2)$ that satisfy (\ref{EquCons}) form a monotonically decreasing curve in the plain of $(\mu_1,\mu_2)$, since by increasing both $\mu_1$ and $\mu_2$, we simultaneously decrease  the left hand side and increase the right hand side of (\ref{EquCons}). Moreover, increasing $\mu_1$ and decreasing  $\mu_2$ results in decreasing $\bar{R}_{1r}$ and $\bar{R}_{r1}$, and increasing $\bar{R}_{2r}$ and $\bar{R}_{r2}$. Therefore, there exists a unique $(\mu_1,\mu_2)$ which satisfies $\mathrm{C1}$ or $\mathrm{C2}$ in (\ref{FixProbMin}), otherwise, either $\mu_1$ or $\mu_2$ reaches a boundary value, i.e., zero or one, which is discussed in the following.

\noindent
\textbf{b)} If $\mu_1 = 1$ and $0<\mu_2<1$, then from (\ref{MET}), we obtain 
\begin{IEEEeqnarray}{rll}
    \Lambda_2(i) &=& \Lambda_3(i) =(1-\mu_2)C_{2r}(i)  \geq \Lambda_1(i) \IEEEyesnumber\IEEEyessubnumber \\
\Lambda_6(i) &=& \mu_2 C_{r1}(i)+C_{r2}(i)\geq \max \{\Lambda_4(i),\Lambda_5(i)\}. \IEEEyessubnumber
\end{IEEEeqnarray}
The probability that the above expressions hold with equality is zero. Therefore, we can set  $\mathcal{I}_k(i)=0,\,\,k=1,4,5$ and $\mathcal{I}_6(i)=1$, $\forall i$. Moreover, we observe that constraint $\mathrm{C2}$ in (\ref{FixProbMin}), i.e., $\bar{R}_{2r}=\bar{R}_{r1}$, only depends on $\mu_2$, and as $\mu_2$ decreases, $\bar{R}_{2r}$ increases and $\bar{R}_{r1}$ decreases. Since if $P_2\leq P_r$, we obtain $\Lambda_2(i)=\Lambda_3(i)<C_{2r}(i)\leq C_{r2}(i)<\Lambda_6(i)$, a sufficient and necessary condition for the existence of a $\mu_2$ such that   constraint $\mathrm{C2}$ in (\ref{FixProbMin}) holds is $P_2 > P_r$. Since both thresholds are obtained, we could determine when to select mode $\mathcal{M}_6$ or $\{\mathcal{M}_2,\mathcal{M}_3\}$. Next, we need to determine the optimal policy to select between $\mathcal{M}_2$ and $\mathcal{M}_3$ if for some channel realizations, we obtain $ \Lambda_2(i) = \Lambda_3(i) > \Lambda_6(i)$ such that constraint $\mathrm{C1}$ in (\ref{FixProbMin}) holds and the maximum sum rate is achieved. Since the maximum sum rate is fixed to $E\left\{q_6(C_{r1}+C_{r2})\right\}$ and it is only a function of $\mu_2$, different strategies of satisfying $\mathrm{C1}$ in (\ref{FixProbMin}) yield the same sum rate. Thus, we consider a simple probabilistic approach for choosing ${\cal M}_3$ with probability $p_1$ and ${\cal M}_2$ with probability $1-p_1$. To implement the probabilistic strategy, we set $\mathcal{I}_3(i)=1-\mathcal{I}_2(i)=X_1(i)$ where $X_1(i)$ is the outcome of a coin flip with probability $p_1=\frac{E\{q_6C_{r2}\}}{E\left\{(1-q_6)[C_r-C_{2r}]\right\}}$ where $p_1$ is chosen such that constraint $\mathrm{C1}$ in (\ref{FixProbMin}) is satisfied.

\noindent
\textbf{c)} If $0<\mu_1 <1$ and $\mu_2=0$, then from (\ref{MET}), we obtain 
\begin{IEEEeqnarray}{rll}
    \Lambda_3(i) &=& (1-\mu_1)[C_r(i)-C_{2r}(i)]+C_{2r}(i) \nonumber \\ 
&\geq& \max \{\Lambda_1(i),\Lambda_2(i)\} \IEEEyesnumber\IEEEyessubnumber \\
\Lambda_5(i) &=& \Lambda_6(i) = \mu_1 C_{r2}(i) \geq \Lambda_4(i).  \IEEEyessubnumber
\end{IEEEeqnarray}
Similar to Case b, the probability that the above expressions hold with equality is zero. Therefore, we set $\mathcal{I}_k(i)=0,\,\,k=1,2,4,\mathcal{I}_3(i)=1$, and $\mathcal{I}_6(i)=1-\mathcal{I}_5(i)=X_2(i)$ with $p_2=\frac{E\{q_3C_{2r}\}}{E\left\{(1-q_3)C_{r1}\right\}}$ for $\forall i$. Moreover, the necessary condition for this case is that $P_2 < P_r$ and the sum rate is $E\{q_3 C_r\}$.

\noindent
\textbf{d)} If $\mu_1 = 1$ and $\mu_2=0$, then from (\ref{MET}), we obtain 
\begin{IEEEeqnarray}{rll}
    \Lambda_2(i) &=& \Lambda_3(i) = C_{2r}(i)  \geq \Lambda_1(i)  \IEEEyesnumber\IEEEyessubnumber \\
\Lambda_5(i) &=& \Lambda_6(i) = C_{r2}(i)  \geq \Lambda_4(i).  \IEEEyessubnumber
\end{IEEEeqnarray}
The probability that the above expressions hold with equality is zero and therefore, we can set $\mathcal{I}_k(i)=0, \,\, k=1,4$. Moreover, one can easily conclude that the necessary condition for this case is $P_2=P_r$, otherwise, only one of $\{{\cal M}_2,{\cal M}_3\}$ or $\{{\cal M}_5,{\cal M}_6\}$ is selected in all  time slots which violates constraints $\mathrm{C1}$ and $\mathrm{C2}$ in (\ref{FixProbMin}). Thus, we obtain $\Lambda_2(i)=\Lambda_3(i)=\Lambda_5(i)=\Lambda_6(i), \,\,\forall i$. Moreover, considering that $\mathrm{C1}$ and $\mathrm{C2}$ in (\ref{FixProbMin}) hold, the sum rate is $\bar{R}_{2r}+\bar{R}_{r2}= E\{(q_2+q_3)C_{2r}\}+E\{(q_5+q_6)C_{r2}\}=E\{C_{2r}\}=E\{C_{r2}\}$. Moreover, the maximum sum rate is achieved by a probabilistic approach as follows
\begin{IEEEeqnarray}{rll}
    {\begin{cases}
\mathcal{I}_2(i)= X_3(i) [1-X_4(i)]\\
\mathcal{I}_3(i)= X_3(i) X_4(i) \\
\mathcal{I}_5(i)=[1-X_3(i)] [1-X_5(i)]\\
\mathcal{I}_6(i)=[1-X_3(i)] X_5(i)
\end{cases}} \IEEEyesnumber
\end{IEEEeqnarray}
where the coin flip probabilities $p_4=\frac{1-p_3}{p_3} \frac{E\{C_{r2}\}}{E\{C_{r}-C_{2r}\}}$ and $p_5=\frac{p_3}{1-p_3} \frac{E\{C_{2r}\}}{E\{C_{r1}\}}$ are chosen such that constraints  $\mathrm{C1}$ or $\mathrm{C2}$ in (\ref{FixProbMin}) are satisfied, respectively, and probability $p_3$ is chosen to achieve the optimal sum rate. However, for the assumed values of $p_4$ and $p_5$, the sum rate is $\bar{R}_{12}+\bar{R}_{21}=E\{C_{2r}\}=E\{C_{r2}\}$ which is indeed the maximum sum rate. Therefore, any value of $p_3$ which results in an acceptable value of $p_4$ and $p_5$, i.e., $0\leq p_4,p_5\leq 1$, is optimal. From this condition, we obtain $\frac{E\{C_{r2}\}}{E\{C_{r}\}} \leq p_3 \leq \frac{E\{C_{r1}\}}{E\{C_{r1}+C_{2r}\}}$, where $p_3 = \frac{E\{C_{r2}\}}{E\{C_{r}\}}$ and $p_3 = \frac{E\{C_{r1}\}}{E\{C_{r1}+C_{2r}\}}$ lead to $p_4=1$ and $p_5=1$, respectively.

\noindent
\textbf{Case 3:} If $t(i)=1$, then considering (\ref{Complementary Slackness}d),  we obtain $\phi_0(i)=0$ and from (\ref{Dual Feasibility Condition}c), we obtain $\phi_1(i)\geq 0$. Combining these results with (\ref{Stationary t}), we obtained $\mu_1\leq\mu_2$. Following a similar procedure as for Case 2, we obtain similar results.

\subsection{SNR Regions}

We obtained different selection policies and the necessary conditions of optimality for each of them. Now, we determine which of the derived selection policies will be used based on the statistics of the channels and the powers of the nodes. To this end, we define three mutually exclusive SNR regions $\mathcal{R}_1,\mathcal{R}_2,$ and $\mathcal{R}_0$. Each of the necessary conditions is realized only in one of these SNR regions and since the SNR regions are mutually exclusive and the necessary conditions are mutually exclusive too, the selection policy for each necessary condition is optimal for the corresponding SNR region. 

First, we observe that the multiple access and broadcast modes are selected with non-zero probability in all the discussed cases. However, the point-to-point modes are selected only when thresholds $\mu_1$ and $\mu_2$ have a boundary value, i.e., zero or one. By this observation, we define the SNR regions $\mathcal{R}_1,\mathcal{R}_2$ and $\mathcal{R}_0$ based on the utilization of point-to-point modes as follows:

\noindent
\textbf{SNR Region $\mathcal{R}_0$:} In this SNR region, none of the point-to-point modes are selected. Moreover, the necessary conditions for the thresholds are $0<\mu_1<1$ and $0<\mu_2<1$.

\noindent
\textbf{SNR Region $\mathcal{R}_1$:} In this SNR region, one or both of point-to-point modes $\mathcal{M}_2$ and $\mathcal{M}_5$ are selected. Moreover, the necessary condition for the thresholds is $\mu_1>\mu_2$ where at least one of the thresholds has a boundary value. Specifically, in this SNR region, if $P_2>P_r$, we obtain $\mu_1=1$ and $0<\mu_2<1$ and mode $\mathcal{M}_2$ is selected in some time slots, whereas if $P_2<P_r$, we obtain $0<\mu_1<1$ and $\mu_2=0$ and mode $\mathcal{M}_5$ is selected in some time slots, and finally if $P_2=P_r$, we obtain $\mu_1=1$ and $\mu_2=0$ and both modes $\mathcal{M}_2$ and $\mathcal{M}_5$ can be selected in some time slots. 

\noindent
\textbf{SNR Region $\mathcal{R}_2$:} In this SNR region, one or both of the point-to-point modes $\mathcal{M}_1$ and $\mathcal{M}_4$ is selected. Moreover, the necessary condition for the thresholds is $\mu_1<\mu_2$ where at least one of the thresholds has a boundary value. Similar to SNR region $\mathcal{R}_1$, in this SNR region, if $P_1>P_r$, we obtain $0<\mu_1<1$ and $\mu_2=1$ and mode $\mathcal{M}_1$ is selected in some time slots, whereas if $P_1<P_r$, we obtain $\mu_1=0$ and $0<\mu_2<1$ and mode $\mathcal{M}_4$ is selected in some time slots, and finally if $P_2=P_r$, we obtain $\mu_1=0$ and $\mu_2=1$ and both modes $\mathcal{M}_1$ and $\mathcal{M}_4$ can be selected in some time slots.  

Roughly speaking, in SNR region $\mathcal{R}_1$, the user 1-to-relay link is much stronger than the user 2-to-relay link, and one or both of the point-to-point modes between the relay and user 2 are selected. Similarly, in SNR region $\mathcal{R}_2$, the user 2-to-relay link is much stronger than the user 1-to-relay link, and one or both of the point-to-point modes between the relay and user 1 are selected. However, in SNR $\mathcal{R}_0$, the quality of user 1-to-relay link is similar to that of the user 2-to-relay link, and none of the point-to-point modes is selected. To easily distinguish the SNR regions based on the channel statistics, we investigate the transition conditions between the regions. We note that as the quality of the user 1-to-relay link improves compared to the user 2-to-relay link, the SNR region changes from $\mathcal{R}_2$ to $\mathcal{R}_0$ and then from $\mathcal{R}_0$ to $\mathcal{R}_1$. Therefore, the transition condition between  SNR regions $\mathcal{R}_1$ and $\mathcal{R}_0$ occurs when the coin flip probabilities are equal to one, i.e., $p_1,p_2,p_4,p_5\overset{ \mathcal{R}_0}{\underset{ \mathcal{R}_1}{\gtrless}} 1$. Then, based on whether $P_r$ is larger, smaller or equal to $P_2$, we use the optimal strategy introduced in Case b, c, and d, respectively. Therefore, the SNR region is $\mathcal{R}_1$ if and only if
\begin{IEEEeqnarray}{lll}\label{Region1}
    &\left\{ P_2>P_r  \,\,  \mathrm{AND}  \,\, \omega_1<1\right\} \nonumber\\
 \mathrm{OR}  \,\,&\left\{ P_2<P_r  \,\,  \mathrm{AND}  \,\, \omega_2<1\right\} \nonumber\\
 \mathrm{OR}  \,\,&\{P_2=P_r \,\,  \mathrm{AND}  \,\, \omega_3^{l} <\omega_3^{u} \} \quad \IEEEyesnumber 
\end{IEEEeqnarray}
where $\omega_1=\frac{E\{qC_{r2}\}}{E\left\{(1-q)[C_r-C_{2r}]\right\}}, \omega_2=\frac{E\{qC_{2r}\}}{E\left\{(1-q)C_{r1}\right\}}, \omega_3^{l} = \frac{E\{C_{r2}\}}{E\{C_{r}\}}$, and $\omega_3^{u}=\frac{E\{C_{r1}\}}{E\{C_{r1}+C_{2r}\}}$ and $q(i)$ is defined as
\begin{IEEEeqnarray}{rl}
   q(i)=\begin{cases}
1, &\mathrm{if}\,\,\,\,\, \left\{ P_2>P_r \,\, \mathrm{AND}\,\, \big[\Lambda_3 \leq \Lambda_6\big]_{\mu_1=1,\mu_2=\mu_2^*}^{t(i)=0,\,\,\forall i} \right\}\\  &\mathrm{OR} \left\{ P_2<P_r \,\, \mathrm{AND}\,\, \big[\Lambda_3 \geq \Lambda_6\big]_{\mu_1=\mu_1^*,\mu_2=0}^{t(i)=0,\,\,\forall i} \right\}
\\0,  & \mathrm{otherwise}
\end{cases} \IEEEeqnarraynumspace
\end{IEEEeqnarray}
where if $P_2>P_r$, the threshold $\mu_2^*$ is chosen such that $E\{(1-q)C_{2r}\}=E\{qC_{r1}\}$ holds, whereas, if $P_2<P_r$, the threshold $\mu_1^*$ is chosen such that $E\{q[C_r-C_{2r}]\}=E\{(1-q)C_{r2}\}$ holds. In a similar manner, one can distinguish between SNR region $\mathcal{R}_2$ and $\mathcal{R}_0$. We note that all conditions of optimality  for these three exclusive SNR regions are obtained in this section and the thresholds $\mu_1$ and $\mu_2$ which resulted in the maximum sum rate are shown to be unique. Thus the proposed solution is optimal. This completes the proof.


\section{Proof of Binary Relaxation}
\label{AppBinRelax}

In this appendix, we prove that the optimal solution of the problem with the relaxed constraints, $0\leq  q_k(i)\leq 1$, can also be achieved if the values of $q_k(i)$ are binary, i.e., if $q_k(i)\in\{0, 1\}$. Therefore, the binary relaxation does not change the maximum sum rate. We start with noting that if one of the $q_k(i),\,\,k=1,\dots,6$, adopts a non-binary value in the optimal solution, then in order for constraint $\mathrm{C3}$ in (\ref{FixProb}) to be satisfied, there has to be at least one other non-binary selection variable in that time slot. Assuming that the mode indices of the non-binary selection variables are $k'$ and $k''$ in the $i$-th time slot, we obtain $\alpha_k(i)=0,\,\,k = 1,\dots,6$ from (\ref{Complementary Slackness}a), $ \beta_{k'}(i)=0$ and $ \beta_{k''}(i)=0$  from (\ref{Complementary Slackness}b). Then, by substituting these values into (\ref{Stationary Mode}), we obtain
\begin{IEEEeqnarray}{lll}\label{BinRelax}
    \lambda(i) = \Lambda_{k'}(i)  \IEEEyesnumber\IEEEyessubnumber  \\
    \lambda(i)= \Lambda_{k''}(i)\IEEEyessubnumber  \\
   \lambda(i)-\beta_3(i) =  \Lambda_k(i), \quad k\neq k', k''. \IEEEyessubnumber
\end{IEEEeqnarray}
From (\ref{BinRelax}a) and  (\ref{BinRelax}b), we obtain $\Lambda_{k'}(i)=\Lambda_{k''}(i)$ and by subtracting (\ref{BinRelax}a) and  (\ref{BinRelax}b) from (\ref{BinRelax}c), we obtain 
\begin{IEEEeqnarray}{rCl}
    \Lambda_{k'}(i) - \Lambda_k(i) &=& \beta_k(i), \quad \quad k\neq k', k'' \IEEEyesnumber \IEEEyessubnumber  \\
\Lambda_{k''}(i) - \Lambda_k(i) &=& \beta_k(i), \quad \quad k\neq k', k''. \IEEEyessubnumber
\end{IEEEeqnarray}
From the dual feasibility condition given in (\ref{Dual Feasibility Condition}b), we have $\beta_k(i)\geq 0$ which leads to $\Lambda_{k'}(i)=\Lambda_{k''}(i)\geq \Lambda_k(i)$. However, as a result of the randomness of the time-continuous channel gains, $\Pr\{\Lambda_{k'}(i)=\Lambda_{k''}(i)\} > 0$ holds for some transmission modes $\mathcal{M}_{k'}$ and $\mathcal{M}_{k''}$ if and only if the optimal value of either $\mu_1$ or $\mu_2$ is at a boundary, i.e., $\mu_1=0,1$ or $\mu_2=0,1$. In these cases, we observe that $\Lambda_{k'}(i)=\Lambda_{k''}(i)$, $\forall i$, and we have shown in Appendix \ref{AppKKT} that for $\mu_1=0,1$ and $\mu_2=0,1$, we can obtain the optimal sum rate by using a probabilistic approach to choose ${\cal M}_{k'}$ with probability $p$ and ${\cal M}_{k''}$ with probability $1-p$, which results in a binary value for the selection variables $q_{k'}(i)$ and $q_{k''}(i)$. This completes the proof.


\section{Threshold Regions}
\label{AppMURegion}

In this appendix, we determine the interval which contains the optimal values of $\mu_1, \mu_2$. Fig. \ref{FigMURegion} represents the set of candidate selection modes for all channel realizations in the space of $(\mu_1,\mu_2)$.  These candidate modes are obtained based on the necessary selection condition introduced in (\ref{OptMet}) as:
\begin{figure}
\centering
\psfrag{M1}[c][c][0.75]{$\mu_1$}
\psfrag{M2}[c][c][0.75]{$\mu_2$}
\psfrag{Mu1}[c][c][0.75]{$\mu_1 = 1$}
\psfrag{M21}[c][c][0.75]{$\,\,\mu_2=1$}
\psfrag{S1}[c][c][0.75]{$\big\{\mathcal{M}_1,\mathcal{M}_4\big\}$}
\psfrag{S2}[c][c][0.75]{$\big\{\mathcal{M}_3,\mathcal{M}_4\big\}$}
\psfrag{S3}[c][c][0.75]{$\big\{\mathcal{M}_3\big\}$}
\psfrag{S4}[c][c][0.75]{$\big\{\mathcal{M}_3,\mathcal{M}_5\big\}$}
\psfrag{S5}[c][c][0.75]{$\big\{\mathcal{M}_3,\mathcal{M}_6\big\}$}
\psfrag{S6}[c][c][0.75]{$\big\{\mathcal{M}_1,\mathcal{M}_6\big\}$}
\psfrag{S7}[c][c][0.75]{$\big\{\mathcal{M}_6\big\}$}
\psfrag{S8}[c][c][0.75]{$\big\{\mathcal{M}_2,\mathcal{M}_6\big\}$}
\psfrag{S9}[c][c][0.75]{$\big\{\mathcal{M}_2,\mathcal{M}_5\big\}$}
\includegraphics[width=3 in]{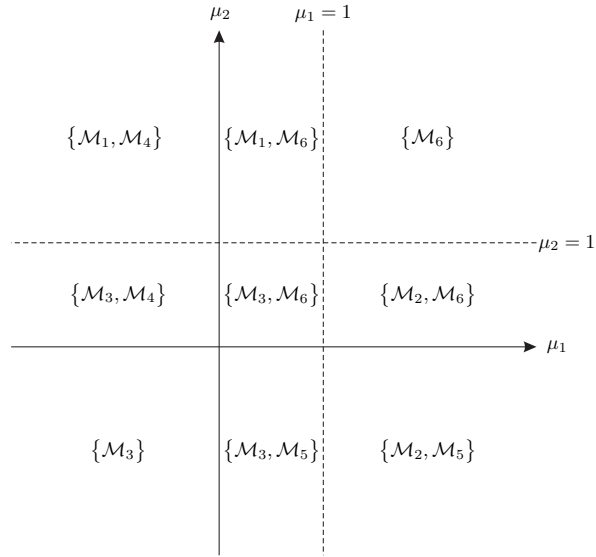}
\caption{Modes with non-negative powers in the space of $(\mu_1,\mu_2)$.}
\label{FigMURegion}
\end{figure}

\noindent
\textbf{1)} If $\mu_1< 1$ and $\mu_2> 1$, then $\mathcal{M}_1$ is a candidate for selection since we obtain $\Lambda_1(i)\geq \max\{\Lambda_2(i),\Lambda_3(i)\}$ from (\ref{MET}).

\noindent
\textbf{2)} If $\mu_1> 1$ and $\mu_2< 1$, then $\mathcal{M}_2$ is a candidate for selection since we obtain $\Lambda_2(i)\geq \max\{\Lambda_1(i),\Lambda_3(i)\}$ from (\ref{MET}).

\noindent
\textbf{3)} If $\mu_1< 1$ and $\mu_2< 1$, then, based on (\ref{MU12}a) and (\ref{M25}a), $\mathcal{M}_3$ is a candidate for selection since we obtain  $\Lambda_3(i)\geq \max\{\Lambda_1(i),\Lambda_2(i)\}$ from (\ref{MET}).

\noindent
\textbf{4)} If $\mu_1< 0$ and $\mu_2> 0$, then $\mathcal{M}_4$ is a candidate for selection since we obtain $\Lambda_4(i)\geq \max\{\Lambda_5(i),\Lambda_6(i)\}$ from (\ref{MET}).

\noindent
\textbf{5)} If $\mu_1> 0$ and $\mu_2< 0$, then $\mathcal{M}_5$ is a candidate  for selection since we obtain $\Lambda_5(i)\geq \max\{\Lambda_4(i),\Lambda_6(i)\}$ from (\ref{MET}).

\noindent
\textbf{6)} If $\mu_1> 0$ and $\mu_2> 0$, then $\mathcal{M}_6$ is a candidate for selection since we obtain $\Lambda_6(i)\geq \max\{\Lambda_4(i),\Lambda_5(i)\}$ from (\ref{MET}).

It is straightforward to see that all sets in Fig.~\ref{FigMURegion} contradict one or both constraints  $\mathrm{C1}$ and $\mathrm{C2}$ in (\ref{FixProb}), except $\{\mathcal{M}_3,\mathcal{M}_6\}$, which corresponds to $0\leq \mu_1 \leq 1$ and $0\leq \mu_2 \leq 1$. Note that if $\mu_1$ and $\mu_2$ are at the boundaries, i.e., $\mu_1=0,1$ and $\mu_2=0,1$, then the neighbor modes are also candidates for selection. 

We note that $\mu_1$ and $\mu_2$ cannot be simultaneously equal to zero since this results in $\Lambda_{4}(i)=\Lambda_{5}(i)=\Lambda_{6}(i)=0$, $\forall i$, and thus only the transmission from the users to the relay will be adopted and this contradicts constraints $\mathrm{C1}$ and $\mathrm{C2}$ in (\ref{FixProb}). Similarly,  $\mu_1$ and $\mu_2$ cannot be simultaneously equal to one since it also contradicts constraints $\mathrm{C1}$ and $\mathrm{C2}$ in (\ref{FixProb}).




\begin{thebibliography}{10}
\providecommand{\url}[1]{#1}
\csname url@samestyle\endcsname
\providecommand{\newblock}{\relax}
\providecommand{\bibinfo}[2]{#2}
\providecommand{\BIBentrySTDinterwordspacing}{\spaceskip=0pt\relax}
\providecommand{\BIBentryALTinterwordstretchfactor}{4}
\providecommand{\BIBentryALTinterwordspacing}{\spaceskip=\fontdimen2\font plus
\BIBentryALTinterwordstretchfactor\fontdimen3\font minus
  \fontdimen4\font\relax}
\providecommand{\BIBforeignlanguage}[2]{{%
\expandafter\ifx\csname l@#1\endcsname\relax
\typeout{** WARNING: IEEEtran.bst: No hyphenation pattern has been}%
\typeout{** loaded for the language `#1'. Using the pattern for}%
\typeout{** the default language instead.}%
\else
\language=\csname l@#1\endcsname
\fi
#2}}
\providecommand{\BIBdecl}{\relax}
\BIBdecl

\bibitem{Tarokh}
S.~J. Kim, N.~Devroye, P.~Mitran, and V.~Tarokh, ``{Achievable Rate Regions and
  Performance Comparison of Half Duplex Bi-Directional Relaying Protocols},''
  \emph{IEEE Trans. Inf. Theory}, vol.~57, no.~10, pp. 6405 --6418, Oct. 2011.

\bibitem{TDBC}
Y.~Wu, P.~A. Chou, and S.-Y. Kung, ``{Information Exchange in Wireless Networks
  with Network Coding and Physical-Layer Broadcast},'' in \emph{Proc. 39th Ann.
  Conf. Inf. Sci. Syst.}, March 2005.

\bibitem{MABC}
P.~Popovski and H.~Yomo, ``{Bi-directional Amplification of Throughput in a
  Wireless Multi-Hop Network},'' in \emph{Proc. IEEE VTC}, vol.~2, May 2006,
  pp. 588--593.

\bibitem{BocheIT}
T.~Oechtering, C.~Schnurr, I.~Bjelakovic, and H.~Boche, ``{Broadcast Capacity
  Region of Two-Phase Bidirectional Relaying},'' \emph{IEEE Trans. Inf.
  Theory}, vol.~54, no.~1, pp. 454 --458, Jan. 2008.

\bibitem{Cover}
T.~M. Cover and J.~A. Thomas, \emph{Elements of Information Theory}.\hskip 1em
  plus 0.5em minus 0.4em\relax Wiley, John and Sons, Incorporated, 1991.

\bibitem{BochePIMRC}
R.~F. Wyrembelski, T.~J. Oechtering, and H.~Boche, ``{Decode-and-Forward
  Strategies for Bidirectional Relaying},'' in \emph{Proc. IEEE PIMRC}, Sept.
  2008, pp. 1 --6.

\bibitem{PopovskiICC}
P.~Popovski and H.~Yomo, ``{Physical Network Coding in Two-Way Wireless Relay
  Channels},'' in \emph{Proc. IEEE ICC}, June 2007, pp. 707 --712.

\bibitem{PopovskiLetter}
H.~Liu, P.~Popovski, E.~de~Carvalho, and Y.~Zhao, ``{Sum-Rate Optimization in a
  Two-Way Relay Network with Buffering},'' \emph{IEEE Commun. Let.}, vol.~17,
  no.~1, pp. 95 --98, Jan. 2013.

\bibitem{NikolaJSAC}
N.~Zlatanov, R.~Schober, and P.~Popovski, ``{Buffer-Aided Relaying with
  Adaptive Link Selection},'' \emph{IEEE J. Select. Areas Commun.}, vol.~31,
  no.~8, pp. 1 --13, Aug. 2013.

\bibitem{NikolaTIT}
\BIBentryALTinterwordspacing
N.~Zlatanov and R.~Schober, ``{Capacity of the State-Dependent Half-Duplex
  Relay Channel Without Source-Destination Link},'' \emph{Submitted IEEE
  Transactions on Information Theory}, 2013. [Online]. Available:
  \url{http://arxiv.org/abs/1302.3777}
\BIBentrySTDinterwordspacing

\end{thebibliography}
\end{document}